%% file: report.tex
\documentclass{article}
\usepackage[verbose=true,letterpaper]{geometry}
\AtBeginDocument{
	\newgeometry{
		textheight=9in,
		textwidth=6.5in,
		top=1in,
		headheight=14pt,
		headsep=25pt,
		footskip=30pt
	} 
}

\usepackage[utf8]{inputenc}

\usepackage{amsmath, amsthm, amssymb}
\usepackage{thmtools, thm-restate}
\usepackage{bbm}
\usepackage{float}
\usepackage{enumerate,enumitem}
\usepackage{color}
\usepackage{mathtools}
\mathtoolsset{showonlyrefs}
\usepackage{ifthen}
\usepackage{xr} 
\allowdisplaybreaks 

\usepackage{pgfplots,tikz}
\pgfplotsset{compat=1.16} 
\usepackage{graphicx}
\usepackage{cite}

\definecolor{ColorLink}{rgb}{0 .4 .6}
\definecolor{ColorCite}{rgb}{0 .5 .3}
\usepackage[plainpages=false,pdfpagelabels,colorlinks=true,linkcolor=ColorLink,citecolor=ColorCite]{hyperref}
\hypersetup{pdfauthor={Name}}

\newtheorem{theorem}{Theorem}
\newtheorem{lemma}{Lemma}
\newtheorem{definition}{Definition}

\newtheorem{remark}{Remark}

\definecolor{BoxRed}{rgb}{1 .8 .8}
\definecolor{BoxLightRed}{rgb}{1 .9 .9}
\definecolor{BoxBlue}{rgb}{.8 .8 1}
\definecolor{BoxLightBlue}{rgb}{.9 .9 1}
\definecolor{BoxCyan}{rgb}{.9 1 1}
\definecolor{BoxGreen}{rgb}{.8 1 .8}
\definecolor{BoxDarkGreen}{rgb}{.6 .9 .6}
\definecolor{BoxYellow}{rgb}{1 .95 .65}
\definecolor{BoxDarkYellow}{rgb}{.95 .85 .5}
\definecolor{BoxOrange}{rgb}{1 .8 0}
\definecolor{BoxPurple}{rgb}{.87 .79 .85}
\definecolor{BoxPink}{rgb}{1 .8 1}
\definecolor{LineGray}{rgb}{.4 .4 .4}
\definecolor{AxisGray}{rgb}{.3 .3 .3}
\definecolor{LineBlue}{rgb}{.2 .6 .9}
\definecolor{LineLightBlue}{rgb}{.4 .8 .95}
\definecolor{LineCyan}{rgb}{.5 .7 .7}
\definecolor{LineRed}{rgb}{.9 .3 .3}
\definecolor{LineLightRed}{rgb}{.95 .4 .4}
\definecolor{LineGreen}{rgb}{.3 .8 .3}
\definecolor{LineDarkGreen}{rgb}{.2 .6 .2}
\definecolor{LineYellow}{rgb}{9 .8 0}
\definecolor{LineDarkYellow}{rgb}{.9 .7 .2}
\definecolor{LineOrange}{rgb}{.9 .6 .1}
\definecolor{LinePurple}{rgb}{.8 .3 .8}
\definecolor{LinePink}{rgb}{.8 .5 .8}

\usepackage{relsize} 

\newcommand{\floor}[1]{\lfloor #1 \rfloor}

\DeclareMathOperator*{\argmax}{arg\,max}

\newcommand\newsubcommand[3]{\newcommand#1{#2\sc@sub{#3}}}
\def\sc@sub#1{\def\sc@thesub{#1}\@ifnextchar_{\sc@mergesubs}{_{\sc@thesub}}}
\def\sc@mergesubs_#1{_{\sc@thesub#1}}

\newcommand{\expe}{\mathbb{E}}

\newcommand{\half}{\frac{1}{2}}
\newcommand*{\dif}{{\mathop{}\!\mathrm{d}}}
\newcommand{\norm}[1]{\left\lVert#1\right\rVert}
\newcommand{\abs}[1]{\left\lvert#1\right\rvert}
\newcommand{\sign}[1]{\mathrm{sign}(#1)}

\def\P{\mathbb{P}}
\def\RR{\mathbb{R}}
\def\id{\mathbf{I}}

\def\zeros{\mathbf{0}}

\def\iid{\stackrel{\mathrm{iid}}{\sim}}

\newenvironment{smallbmatrix}{
	\bigl[ \begin{smallmatrix} }{%
	\end{smallmatrix} \bigr]}

\makeatletter
\def\thanks#1{\protected@xdef\@thanks{\@thanks
        \protect\footnotetext{#1}}}
\makeatother

\title{Detecting Correlated Gaussian Databases
	\thanks{The authors were supported by NSF Grant CCF-1955981. This work was presented in part at the 2022 IEEE International Symposium on Information Theory~\cite{k2022databases}.}
	\thanks{Z. K and B. Nazer are with the Department of Electrical and Computer Engineering, Boston University, Boston, MA, email: \texttt{zeynepk, bobak@bu.edu}.}
}
\author{Zeynep K and Bobak Nazer}

 \date{\vspace{-12pt}}

\begin{document}
\maketitle

\begin{abstract}
This paper considers the problem of detecting whether two databases, each consisting of $n$ users with $d$ Gaussian features, are correlated. Under the null hypothesis, the databases are independent. Under the alternate hypothesis, the features are correlated across databases, under an unknown row permutation. A simple test is developed to show that detection is achievable above $\rho^2 \approx \frac{1}{d}$. For the converse, the truncated second moment method is used to establish that detection is impossible below roughly $\rho^2 \approx \frac{1}{d\sqrt{n}}$. These results are compared to the corresponding recovery problem, where the goal is to decode the row permutation, and a converse bound of roughly $\rho^2 \approx 1 - n^{-4/d}$ has been previously shown. For certain choices of parameters, the detection achievability bound outperforms this recovery converse bound, demonstrating that detection can be easier than recovery in this scenario.
\end{abstract}

\input{1-Intro}

\input{2-Problem}
\input{3-Results}

\input{4-Ach}

\input{5-Converse}

\section*{Acknowledgments}
The authors would like to thank Reviewer 2 for our ISIT submission for catching an error in our converse bound, which is corrected here and in the final ISIT paper.

\addcontentsline{toc}{section}{References}
\bibliographystyle{ieeetr}
\bibliography{ref.bib}

\appendix
\renewcommand{\theequation}{\thesection.\arabic{equation}}

\input{app-Converse}
\input{app-Ineqs}

\end{document}

%% file: 1-Intro.tex
\section{Introduction}
Consider the following na\"ive approach to database anonymization: prior to public release, unique identifying information (e.g., names, user IDs) is deleted from a database while other features (e.g., movie ratings) are left unchanged. In the absence of any side information, this approach could be effective for protecting user privacy while providing access to data. However, side information is abundant in the public domain, and the work of Narayanan and Shmatikov~\cite{narayanan2008robust, narayanan2009anonymizing} demonstrated that this na\"ive approach is vulnerable to a de-anonymization attack that exploits the availability of another correlated database with its identifying information intact.

This \textit{database alignment} problem, introduced by Cullina et al.~\cite{cullina2018fundamental}, is a simple probabilistic model of the scenario above. There are two databases, each with $n$ users and $d$ features per user. An unknown, random, uniform permutation matches users in the first database to users in the second database. For a pair of matched database entries, the features are dependent according to a known distribution, and, for unmatched entries, the features are independent. The goal is to \textit{recover} the unknown permutation, i.e., align the databases, and the corresponding information-theoretic question is to characterize the level of dependency required for successful recovery, as a function of $n$ and $d$. The discrete memoryless case was studied in~\cite{cullina2018fundamental}, which derived achievability and converse bounds in terms of mutual information. Subsequent work extended this analysis to the case of Gaussian features~\cite{dai2019database}. 

This paper considers the corresponding \textit{database correlation detection} problem: given two databases, can we determine whether they are correlated? Specifically, we focus on the special case of Gaussian features and formulate the following binary hypothesis test. As above, there are two databases, each with $n$ users and $d$ features per user. Under the alternate hypothesis, the databases are correlated as in the recovery problem with respect to an unknown, random, uniform permutation. Under the null hypothesis, the databases are independent. We propose a simple, efficient test for detecting correlated databases, and use it to derive achievability bounds. We also derive converse results by bounding the (truncated) second moment of the likelihood ratio, which can in turn be linked to the total variation distance between the null and alternate distributions. Although there remains a gap between our achievability and converse bound, we are able to show that, in certain regimes, the achievability bound for detection outperforms the best-known converse for recovery from~\cite{dai2019database}, demonstrating that Gaussian database correlation detection is easier than Gaussian database alignment.

\subsection{Related Work}

Recent efforts have generalized the database alignment work of Cullina et al.~\cite{cullina2018fundamental} along several  directions. As mentioned above, exact recovery of the permutation for correlated Gaussian databases was investigated in~\cite{dai2019database}, and follow-up work extended these bounds to near-exact (i.e., partial) recovery~\cite{dai2020partial}. A typicality-based framework for database alignment was explored in~\cite{shirani2019concentration} while random feature deletions and repetitions were studied in \cite{bakirtas2021columndel} and~\cite{bakirtas2022columnrep}, respectively. The Gaussian database alignment problem is also equivalent to a certain idealized tracking problem studied by \cite{chertkov2010inference, kunisky2021geometricmatching}, inspired by the application of particle tracking to infer the trajectories of objects from sequences of still images. Note that in this setting, the pairings are coupled whereas reconstruction of the planted matching in random bipartite graphs proposed by \cite{moharrami2021planted, ding2021planted} deals with independent random pairings.

Database alignment belongs to the larger family of planted matching problems, where the high-level goal is to recover or detect a hidden combinatorial structure in a high-dimensional setting. While many of these scenarios are known to have pessimistic computational complexity bounds in the worst-case setting, recent efforts have demonstrated that certain scenarios admit efficient algorithms in the average-case setting, see the survey of Wu and Xu~\cite{wu2021statistical} for more details. For instance, consider the graph alignment problem. There are two graphs over $n$ nodes, whose edges are generated in a correlated fashion, and then the node labels are randomly and uniformly permuted. The recovery problem is to use the correlated graphs to find the true node labeling, which is equivalent to a certain quadratic assignment optimization problem. Initial work~\cite{pedarsani2011privacy} on this problem in the information theory literature proposed the correlated Erd$\ddot{\text{o}}$s–R{\'e}nyi graph model with dependent Bernoulli edge pairs. Subsequent work considered the recovery problem for the Gaussian setting~\cite{jiaming2021degree}, specifically correlated Gaussian Wigner matrices. Other recent papers have considered the corresponding detection problem for graph correlation \cite{jiaming2021testinggraphs,mao2021testing} and it is now known~\cite{wu2022settling, ganassali2022sharp} that detecting whether Gaussian graphs are correlated is as difficult as recovering the node labeling. We also mention the extensive work on community recovery and detection for the stochastic block model and refer readers to the survey of Abbe~\cite{abbe2017community} for a finer-grained view.

\subsection{Outline}
The remainder of this paper is organized as follows. In Section~\ref{sec: probstate}, we give formal problem statements for recovery and detection for correlated Gaussian databases. Then, in Section~\ref{sec: mainresults}, we present our achievability and converse bounds for the detection problem, along with the best-known recovery bounds for comparison. Section~\ref{sec: ach-inner} proposes an efficient test and develops achievability results by bounding certain moment generating functions. Section~\ref{sec: converse} presents a converse by bounding the second moment of the likelihood ratio as well as a truncated version. Technical lemmas and proofs are relegated to the appendices.
 \\

%% file: 2-Problem.tex
\section{Problem Statement} \label{sec: probstate}

In this section, we describe a simple Gaussian model of correlated databases as well as the corresponding detection and recovery problems. To start, we set some notational conventions. \\

\noindent\textbf{Notation.} Random column vectors are denoted by capital letters such as $X$ with transpose $X^\mathsf{T}$. A collection of $n$ random vectors is written as $X^n = (X_1,\ldots,X_n)$. We use $X \stackrel{d}{=} Y$ to denote that the random vectors $X$ and $Y$ have the same distribution. The notation $X_1,\ldots,X_n \iid P_X$ denotes that the random vectors $X_1,\ldots,X_n$ are independent and identically distributed (i.i.d.)~according to $P_X$, or, equivalently, $X^n = (X_1,\ldots,X_n) \sim P_X^{\otimes n}$. We write $\expe_{X \sim P_X}$ to denote the expectation with respect to $X$, which is distributed according to $P_X$. We will sometimes simply write $\expe_{X}$ when the distribution is clear from context. 

The $n\times n$ identity matrix is denoted by $\id_n$, the $n\times n$  all-zeros matrix by $\zeros_n$, and the length-$n$. Let $[n] \triangleq \{1,\ldots,n\}$ and define $S_n$ as the set of all permutations over $[n]$. For a given permutation $\sigma \in S_n$, let $\sigma_i$ denote the value to which $\sigma$ maps $i \in [n]$. 

We use $\mathcal{N}(\mu,{\Sigma})$ to represent the multivariate normal distribution with mean vector $\mu$ and covariance matrix $\Sigma$ and $\chi^2_n$ for the chi-squared distribution with $n$ degrees of freedom. If the probability measure $P$ is absolutely continuous with respect to the probability measure $Q$, then we will (with a slight abuse of notation) use $\frac{P}{Q}$ to represent the Radon-Nikodym derivative.

For two real-valued functions $f(x)$ and $g(x)$, we employ the standard asymptotic notation $f(x) = \mathcal{O}(g(x))$ to mean that there exist constants $C > 0$ and $x_0$ such that $\abs{f(x)} \leq C \abs{g(x)}$ for all $x\geq x_0$. Similarly,  $f(x)= \Omega(g(x))$ means that there exist constants $C > 0$ and $x_0$ such that $\abs{f(x)} \geq C \abs{g(x)}$ for all $x\geq x_0$, $f(x) = o(g(x))$ that for every $c > 0$, there exists $x_0$ such that $| f(x) | \leq c g(x)$ for all $x \geq x_0$, and $f(x)= \omega(g(x))$ that for every $c > 0$, there exists $x_0$ such that $| f(x) | \geq c g(x)$ for all $x \geq x_0$.

\subsection{Probabilistic Database Model}

\begin{figure} \label{fig:H0}
\begin{center}
\begin{tikzpicture}[outer sep=auto]
	\tikzset{every node/.style={shape=rectangle,draw, line width = 1.5pt,minimum height = 15pt,minimum width = 150pt}}
	\tikzset{every edge/.style={draw,line width = 1.5pt, color = LineGray}}
	\node [fill=BoxBlue, draw = LineBlue] at (0,3) (X1) {$X_1$};
	\node [fill=BoxGreen, draw = LineGreen] at (0,2.25) (X2) {$X_2$};
	\node [fill=BoxRed, draw = LineRed] at (0,1.5) (X3) {$X_3$};
	\node [draw = none] at (0,.85) (leftdots) {\Large{$\vdots$}};
	\node [fill=BoxYellow, draw = LineYellow] at (0,0) (Xn) {$X_n$};
	\draw [<->,line width = 1.5pt, color = LineGray] (-3,-0.25) -- (-3,3.25) node[above, color = black, pos = 0.5, draw = none, rotate = 90] {$n$ users};
	\draw [<->,line width = 1.5pt, color = LineGray] (-2.65,3.6) -- (2.65,3.6) node[above, color = black, pos = 0.5, draw = none] {$d$ features};
		\node [draw = none] at (0,-.8) (database1) {\Large{Database $X^n$}};

	\node [fill=BoxCyan, draw = LineCyan] at (9,3) (Y1) {$Y_1$};
	\node [fill=BoxPink, draw = LinePink] at (9,2.25) (Y2) {$Y_2$};
	\node [fill=BoxOrange, draw = LineOrange] at (9,1.5) (Y3) {$Y_3$};
	\node [draw = none] at (9,.85) (rightdots) {\Large{$\vdots$}};
	\node [fill=BoxPurple, draw = LinePurple] at (9,0) (Yn) {$Y_n$};
	\draw [<->,line width = 1.5pt, color = LineGray] (12,-0.25) -- (12,3.25) node[below, color = black, pos = 0.5, draw = none, rotate = 90] {$n$ users};
		\draw [<->,line width = 1.5pt, color = LineGray] (6.35,3.6) -- (11.65,3.6) node[above, color = black, pos = 0.5, draw = none] {$d$ features};
		\node [draw = none] at (9,-.8) (database2) {\Large{Database $Y^n$}};

\end{tikzpicture}
\end{center}
\caption{Under the null hypothesis $H_0$ for the detection problem, the database entries $X_1, \ldots, X_n, Y_1,\ldots Y_n \in \mathbb{R}^d$ are generated i.i.d.~according to $P_X = P_Y= \mathcal{N}^{\otimes d}(0,1)$.}
\end{figure}
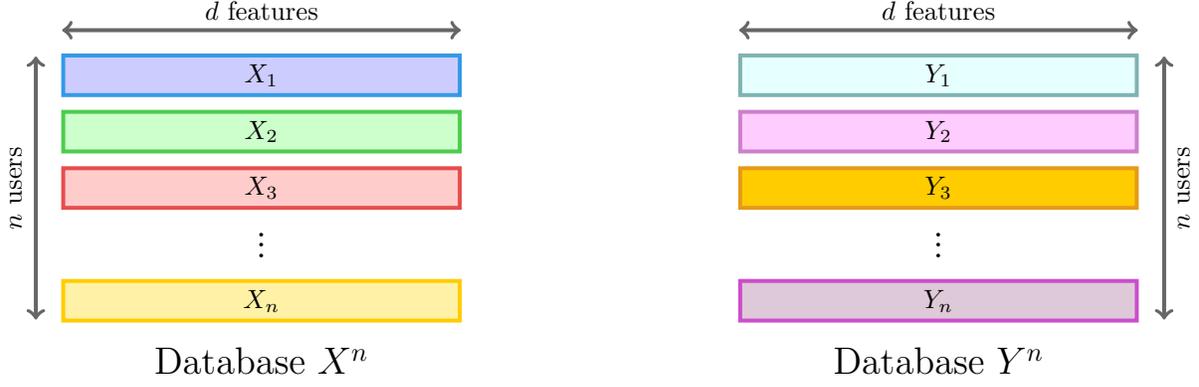

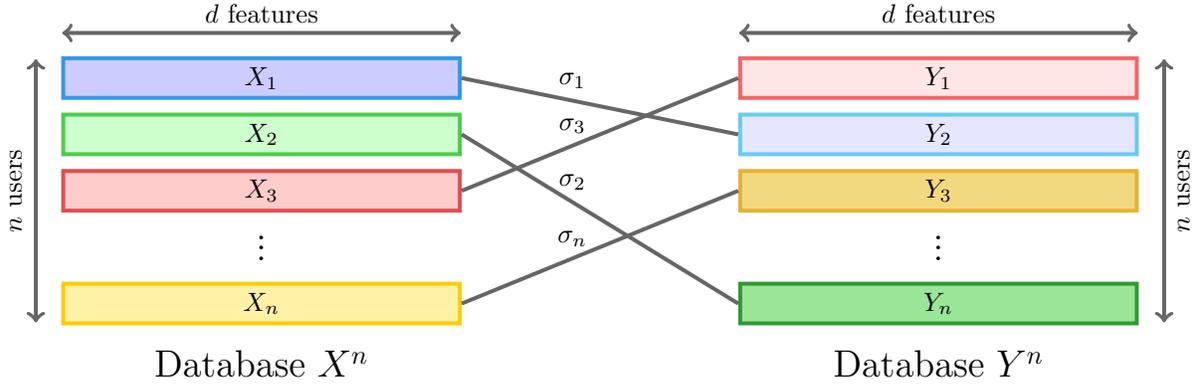
\begin{figure}\label{fig:H1}
\begin{center}
\begin{tikzpicture}[outer sep=auto]
	\tikzset{every node/.style={shape=rectangle,draw, line width = 1.5pt,minimum height = 15pt,minimum width = 150pt}}
	\tikzset{every edge/.style={draw,line width = 1.5pt, color = LineGray}}
	\node [fill=BoxBlue, draw = LineBlue] at (0,3) (X1) {$X_1$};
	\node [fill=BoxGreen, draw = LineGreen] at (0,2.25) (X2) {$X_2$};
	\node [fill=BoxRed, draw = LineRed] at (0,1.5) (X3) {$X_3$};
	\node [draw = none] at (0,.85) (leftdots) {\Large{$\vdots$}};
	\node [fill=BoxYellow, draw = LineYellow] at (0,0) (Xn) {$X_n$};
	\draw [<->,line width = 1.5pt, color = LineGray] (-3,-0.25) -- (-3,3.25) node[above, color = black, pos = 0.5, draw = none, rotate = 90] {$n$ users};
	\draw [<->,line width = 1.5pt, color = LineGray] (-2.65,3.6) -- (2.65,3.6) node[above, color = black, pos = 0.5, draw = none] {$d$ features};
	\node [draw = none] at (0,-.8) (database1) {\Large{Database $X^n$}};

	\node [fill=BoxLightRed, draw = LineLightRed] at (9,3) (Y1) {$Y_1$};
	\node [fill=BoxLightBlue, draw = LineLightBlue] at (9,2.25) (Y2) {$Y_2$};
	\node [fill=BoxDarkYellow, draw = LineDarkYellow] at (9,1.5) (Y3) {$Y_3$};
	\node [draw = none] at (9,.85) (rightdots) {\Large{$\vdots$}};
	\node [fill=BoxDarkGreen, draw = LineDarkGreen] at (9,0) (Yn) {$Y_n$};
	\draw [<->,line width = 1.5pt, color = LineGray] (12,-0.25) -- (12,3.25) node[below, color = black, pos = 0.5, draw = none, rotate = 90] {$n$ users};
		\draw [<->,line width = 1.5pt, color = LineGray] (6.35,3.6) -- (11.65,3.6) node[above, color = black, pos = 0.5, draw = none] {$d$ features};
	\node [draw = none] at (9,-.8) (database2) {\Large{Database $Y^n$}};

	\path (X1.east) edge  node[draw=none,pos = 0.4,above, color = black] (edge1) {$\sigma_1$} (Y2.west)
		(X2.east) edge  node[draw=none,pos = 0.4,above, color = black] (edge2) {$\sigma_2$} (Yn.west)
		(X3.east) edge  node[draw=none,pos = 0.4,above, color = black] (edge3) {$\sigma_3$} (Y1.west)
		(Xn.east) edge  node[draw=none,pos = 0.4,above, color = black] (edgen) {$\sigma_n$} (Y3.west);
\end{tikzpicture}
\end{center}
\caption{Under the alternate hypothesis $H_1$ in the detection problem (or the probabilistic model for the recovery problem), each pair of database entries $(X_i, Y_{\sigma_i}) \in \mathbb{R}^d \times \mathbb{R}^d$ is generated in a dependent fashion for some unknown permutation $\sigma \in S_n$. Specifically, $(X_i, Y_{\sigma_i}) \sim P_{XY} = \mathcal{N}^{\otimes d}\big(\begin{bsmallmatrix}0 \\ 0 \end{bsmallmatrix},\begin{bsmallmatrix}1 & \rho \\ \rho & 1 \end{bsmallmatrix}\big)$. The collection of pairs $(X_1, Y_{\sigma_1}), \ldots, (X_n, Y_{\sigma_n})$ is independent.}
\end{figure}

In our considerations, a \textit{database} is a collection $X^n = (X_1,\ldots,X_n)$ of $n$ i.i.d.~random vectors in $\RR^d$ where $n$ is the number of users (or entries) and $d$ is the number of features. We say that a pair of databases $X^n = (X_1,\ldots,X_n)$ and $Y^n = (Y_1,\ldots,Y_n)$ is \textit{correlated with permutation $\sigma$} if, for some $\sigma \in S_n$, we have that $(X_1,Y_{\sigma_1}),\ldots,(X_n,Y_{\sigma_n}) \iid P_{XY}$ for some joint distribution $P_{XY}$ over $\RR^d \times \RR^d$. We will focus on the special case of \textit{correlated Gaussian databases} where $P_{XY} =\mathcal{N}^{\otimes d}\big(\begin{bsmallmatrix}0 \\ 0 \end{bsmallmatrix},\begin{bsmallmatrix}1 & \rho \\ \rho & 1 \end{bsmallmatrix}\big)$.

\subsection{Detection}
Consider the following binary hypothesis testing problem. Under the \textit{null hypothesis} $H_0$, the Gaussian databases $X^n$ and $Y^n$ are generated independently with $X_1,\ldots,X_n,Y_1,\ldots,Y_n \iid \mathcal{N}(0_d,\id_d)$. See Figure~\ref{fig:H0} for an illustration. (This is equivalent to generating Gaussian correlated databases with $\rho = 0$.) Let $\mathbb{P}_{0}$ denote the resulting distribution over $(X^n,Y^n)$. Under the \textit{alternate hypothesis} $H_1$, the databases $X^n$ and $Y^n$ are correlated with permutation $\sigma$ for some unknown $\sigma \in S_n$ and some known correlation coefficient $\rho \neq 0$. See Figure~\ref{fig:H1} for an illustration and let $\mathbb{P}_{1|\sigma}$ denote the resulting distribution over $(X^n,Y^n)$. Summarizing, 
\begin{align}\label{eq: problem}
	\begin{split}
		H_0:&~~ (X_1,Y_1),\ldots,(X_n,Y_n) \iid \mathcal{N}^{\otimes d}\bigg(\begin{bmatrix}0 \\ 0 \end{bmatrix},\begin{bmatrix}1 & 0 \\ 0 & 1 \end{bmatrix}\bigg)  \\
		H_1:&~~ (X_1,Y_{\sigma_1}),\ldots,(X_n,Y_{\sigma_n}) \iid \mathcal{N}^{\otimes d}\bigg(\begin{bmatrix}0 \\ 0 \end{bmatrix},\begin{bmatrix}1 & \rho \\ \rho & 1 \end{bmatrix}\bigg)
	\end{split}
\end{align} for some permutation $\sigma \in S_n$.  

Given only the databases $X^n$ and $Y^n$ (and not the permutation $\sigma$), a \textit{test} $\phi: \mathbb{R}^{d \times n} \times \mathbb{R}^{d \times n} \rightarrow \{0,1\}$ guesses whether the null, $\phi(X^n,Y^n) = 0$,  or alternate, $\phi(X^n,Y^n) = 1$ hypothesis occurred. For a given $n$ and $d$, the \textit{risk} of a test $\phi$ is 
\begin{align*}
R(\phi) \triangleq  \underbrace{\vphantom{\max_{\sigma \in S_n}} \mathbb{P}_0\big\{\phi(X^n,Y^n) = 1\big\}}_{\text{False Alarm}} + \underbrace{\,\max_{\sigma \in S_n}\mathbb{P}_{1|\sigma} \big\{\phi(X^n,Y^n) = 0\big\}}_{\text{Missed Detection}}
\end{align*}
 and the \textit{minimax risk} is $$R^* \triangleq \inf_{\phi} R(\phi)$$ where the infimum is over all measurable tests. 
 
We will use the shorthand notation $\expe_{0}$ to denote the expectation with respect to $(X^n, Y^n)$ under the null distribution and $\expe_{1|\sigma}$ to denote the expectation with respect to $(X^n, Y^n)$ under the alternate distribution with permutation $\sigma$.
\subsection{Recovery} 
Consider the following recovery problem. A pair of Gaussian correlated databases $X^n$ and $Y^n$ are generated according to 
\begin{align*}
	(X_1,Y_{\sigma_1}),\ldots,(X_n,Y_{\sigma_n}) \iid  \mathcal{N}^{\otimes d}\bigg(\begin{bmatrix}0 \\ 0 \end{bmatrix},\begin{bmatrix}1 & \rho \\ \rho & 1 \end{bmatrix}\bigg)
\end{align*} for some permutation $\sigma \in S_n$. See Figure~\ref{fig:H1} for an illustration. To remain consistent with our detection notation, the resulting distribution is denoted as $\mathbb{P}_{1|\sigma}$. The goal is to decode the permutation $\sigma \in S_n$ given $X^n$ and $Y^n$. Specifically, the \textit{decoder} is a function $\hat{\sigma}: \mathbb{R}^{d \times n} \times \mathbb{R}^{d \times n} \rightarrow S_n$. We define the probability of error 
$$P_{\text{err}}(\hat{\sigma}) \triangleq \max_{\sigma \in S_n} \mathbb{P}_{1|\sigma}\big\{\hat{\sigma}(X^n,Y^n) \neq \sigma\big\}$$
and the \textit{minimax probability of error} $$P^*_{\text{err}} \triangleq \inf_{\hat{\sigma}} P_{\text{err}}(\hat{\sigma})$$ where the infimum is over all measurable decoders.\footnote{Note that prior work~\cite{dai2019database} considers the average-case probability of error with respect to a uniform distribution over permutations, $\sigma \sim \mathrm{Unif}(S_n)$. 
	However, owing to the symmetries of the distribution, these two  models are essentially equivalent, and the results of~\cite{dai2019database} translate directly to minimax bounds.} 

%% file: 3-Results.tex
\section{Main Results} \label{sec: mainresults}
In this section, we present our achievability and converse results for the correlated Gaussian database detection problem. For comparison, we reproduce the main achievability and converse theorems from~\cite{dai2019database} for the correlated Gaussian database recovery problem. 

For our test, we first calculate the sum-of-inner-products statistic 
\begin{align}
T \triangleq \sign{\rho} \sum_{i,j} X_i^\mathsf{T}Y^{\vphantom{\mathsf{T}}}_j \label{eq:statistic}
\end{align} We then compare $T$ to a threshold $t$ as follows:
\begin{align}\label{eq: siptest}
	\phi_{T,t}(X^n,Y^n) &= \begin{dcases} 0, &  \sign{\rho} \sum_{i,j} X_i^\mathsf{T}Y^{\vphantom{\mathsf{T}}}_j < t \\  1, &  \sign{\rho} \sum_{i,j} X_i^\mathsf{T}Y^{\vphantom{\mathsf{T}}}_j \geq t \ . \\[-8pt]
	\end{dcases}
\end{align} Let $\phi_T$ be the optimal test over such  functions, i.e., $$R(\phi_T) = \inf_t R(\phi_{T,t})$$ where the minimizer $t\in (\expe_{0}\,T, ~ \min_{\sigma \in S_n} \expe_{1| \sigma} \,T) = (0,|\rho|nd)$ may depend on $n$, $d$, and $\rho$.

To obtain an upper bound on the minimax risk, we apply the Chernoff bound to the probability of false alarm and the probability missed detection, and select the threshold that balances the resulting exponents. The details of this analysis can be found in Section~\ref{sec: ach-inner}, and culminate in the following theorem.

\begin{theorem}[Detection Achievability]\label{thm:detectionachievable} Let $t = \sqrt{\gamma}\frac{dn}{2}$  with $\gamma \in (0,4\rho^2)$. The risk of the sum-of-inner-products test $\phi_T$ for the binary hypothesis testing problem~\eqref{eq: problem} is upper bounded by
	\begin{align*}
		R(\phi_T) &\leq\,  \min_{\gamma \in (0,4\rho^2)}  \exp\left(-\frac{d}{2} g_{\mathrm{FA}}(\gamma)\right)  + \exp\left(-\frac{d}{2} g_{\text{MD}}(\gamma)\right)   \\&\leq 2 \exp\left(-\frac{d \rho^2}{60}\right) 
		\intertext{ where }
		g_{\text{FA}}(\gamma) &\triangleq \sqrt{1 + \gamma} - 1 - \ln\left(\frac{1 + \sqrt{1 + \gamma}}{2}\right) \ ,\\
		g_{\mathrm{MD}}(\gamma) &\triangleq \frac{1}{1 - \rho^2}\Big( \!\sqrt{ (1-\rho^2)^2 +\gamma} \!- \!\sqrt{\rho^2 \gamma}\Big)  - 1 -  \ln\!\bigg( \frac{1 - \rho^2+\sqrt{ (1-\rho^2)^2 +\gamma }}{2}\bigg)  \ .
	\end{align*} 
	
\end{theorem}

It follows from the simple upper bound $R(\phi_T) \leq 2 \exp(-\frac{d \rho^2}{60})$  that, if $\rho^2 = \omega\big(1/d \big)$, then $R(\phi_T) \to 0$ as $d\to \infty$. Our converse argument relies on a truncated second-moment method and can be found in Section~\ref{sec: converse}. 

\begin{theorem}[Detection Converse]\label{thm:detectionconverse} For $n\geq e^2$,
	if $$\rho^2 = o\bigg(\frac{1}{d\sqrt{n}}\bigg) \; \text{ and } \; d=\Omega(\ln n ) \ ,$$ then	the minimax risk for the binary hypothesis testing problem~\eqref{eq: problem} goes to 1, i.e., $R^* \to 1$ as $d$, $n\to \infty$. 
\end{theorem} 

For the sake of comparison, we now revisit the key results from~\cite{dai2019database} for the recovery problem.

\begin{theorem}[Recovery Achievability, {\cite[Theorem 1]{dai2019database}}] \label{thm:recoveryachievable}
	The probability of error of the maximum likelihood (ML) decoder $\hat{\sigma}_{\text{ML}} = \argmax_{\sigma \in S_n} \mathbb{P}_{1|\sigma}\left(X^n,Y^n\right)$ is upper bounded by
	\begin{align*}
		P_{\text{err}}(\hat{\sigma}_{\text{ML}}) \leq n(1-\rho^2)^{\frac{d}{4}}\frac{ 1-\big(n(1-\rho^2)^{\frac{d}{4}}\big)^n}{1-n(1-\rho^2)^{\frac{d}{4}}}
	\end{align*} 
\end{theorem} Therefore, if $\rho^2 = 1- o(n^{-\frac{4}{d}})$, then the ML decoder returns the true permutation with high probability.

\begin{theorem}[Recovery Converse, {\cite[Theorem 2]{dai2019database}}] \label{thm:recoveryconverse}
	The minimax probability of error is lower bounded by  
	\begin{align*}
		P^*_{\text{err}} \geq 1 - \big( n (1-\rho^2)^{\frac{d}{4}(1+\epsilon(d))}\big)^{-2} - 4\big( n (1-\rho^2)^{\frac{d}{4}(1+\epsilon(d))}\big)^{-1}
	\end{align*} where $\epsilon(d)\to 0$ as $d\to\infty$. \end{theorem} 
Consequently, if $\rho^2 = 1-\omega(n^{-\frac{4}{d}})$, then the probability of error of any decoder is close to $1$.

In Figures~\ref{fig:riskvsd} and~\ref{fig:riskvsn}, we have the plotted upper and lower bounds on the squared correlation coefficient $\rho^2$ required to attain a specified risk $R$ that follow from Theorems~\ref{thm:detectionachievable} and~\ref{thm:detectionconverse}. The quantitative form of the converse result in Theorem~\ref{thm:detectionconverse} follows from Lemmas~\ref{lemma: truncated-first} and~\ref{lemma: truncated-second} by choosing the parameters appropriately. For comparison, we have also plotted the achievable and converse results from the recovery problem from Theorems~\ref{thm:recoveryachievable} and~\ref{thm:recoveryconverse}. Specifically, any recovery scheme can be converted into a detection scheme by first estimating the permutation as $\hat{\sigma}$ and then thresholding the statistic $\mathrm{sign}(\rho) \sum_{i} X_i^T Y_{\hat{\sigma}_i}$. These plots show that, in certain regimes, detection is possible at a lower squared correlation coefficient $\rho^2$ than recovery.

\begin{figure}[b!]
	\centering
	\begin{minipage}{0.485\columnwidth}
		\centering
		\input{riskvsdtikzfig.tex}
		\caption{Upper and lower bounds on the squared correlation coefficient $\rho^2$ needed to attain risk $R = 0.1$ for $n = 10000$ users, with respect to the number of features $d$.}
		\label{fig:riskvsd}
	\end{minipage}\hfill
	\begin{minipage}{0.485\columnwidth}
		\centering
		\input{riskvsntikzfig.tex}
		\caption{Upper and lower bounds on the squared correlation coefficient $\rho^2$ needed to attain risk $R = 0.1$ for $d = 1000$ features, with respect to the number of users $n$.}
		\label{fig:riskvsn}
	\end{minipage}
\end{figure}

\begin{remark}
As discussed above, the particle tracking problem studied in~\cite{kunisky2021geometricmatching} is equivalent to the correlated Gaussian database recovery problem from~\cite{dai2019database}. Specifically, for a given permutation $\sigma$, a  correlated standard   normal pair $(X_i, Y_{\sigma_i})$ can be written as $\rho^{-1}Y_{\sigma_i} = X_i + \sqrt{\rho^{-2}-1}\, Z_i$ where $X_i$ and $Z_i$ are independent standard normal variables and $ \sqrt{\rho^{-2}-1}$ controls the deviation of the noise $Z_i$. In this setting, the achievability bounds from~\cite{kunisky2021geometricmatching} show that, for $d = o(\ln n)$, $\rho^2 = (1+o(n^{-\frac{4}{d}}))^{-1}$ suffices, and, for $d = \omega(\ln n)$, $\rho^{-2} = \big(\frac{1}{4}-\epsilon\big)\mathcal{O}\big(\frac{d}{\ln n}\big)+1$ for some $\epsilon>0$ suffices. These bounds are looser than those from~\cite{dai2019database}, which demonstrated a sharp threshold at $\frac{d}{\ln n} \ln \frac{1}{1-\rho^2}\approx 4$. 
\end{remark}

\begin{remark}
In the graph alignment setting, detection and recovery problems exhibit the same asymptotic behavior, in contrast to the distinct asymptotics of detection and recovery for database alignment. That is, for two graphs over $n$ nodes with edges that are jointly drawn from a $\rho$-correlated bivariate standard normal distribution, the risk of detection and the error probability of recovery have a sharp transition at the threshold $\frac{n \rho^2}{\ln n}\approx 4$. For more information, see \cite[Theorem 1]{jiaming2021testinggraphs}, \cite[Theorem 1]{wu2022settling} and  \cite[Theorem 1]{ganassali2022sharp}. 
\end{remark}

%% file: riskvsdtikzfig.tex
%
%
%
\definecolor{mycolor1}{rgb}{0.00000,1.00000,1.00000}%
\begin{tikzpicture}

\begin{axis}[%
width=2.6in,
height=2.53in,
scale only axis,
separate axis lines,
every outer x axis line/.append style={white!15!black},
every x tick label/.append style={font=\color{white!15!black}},
x tick label style={/pgf/number format/.cd,%
          scaled x ticks = false,
          set thousands separator={},
          fixed},
xmin=0,
xmax=10000,
xlabel={$d$, Number of Features},
every outer y axis line/.append style={white!15!black},
every y tick label/.append style={font=\color{white!15!black}},
ymode=log,
ymin=1e-05,
ymax=1,
yminorticks=true,
ylabel={$\rho{}^\text{2}$, Squared Correlation Coefficient},
title={Number of Users $n = 10000$, Risk $R = 0.1$},
legend style={draw=white!15!black,fill=white,legend cell align=left}
]
\addplot [color=LineGreen,dotted,line width=2.0pt]
  table[row sep=crcr]{%
18.4206807439524	0.930130004882812\\
222.126381136933	0.198029518127441\\
425.832081529913	0.108736038208008\\
629.537781922894	0.0749120712280273\\
833.243482315875	0.0571327209472656\\
1036.94918270886	0.046173095703125\\
1240.65488310184	0.0387411117553711\\
1444.36058349482	0.0333690643310547\\
1648.0662838878	0.0293064117431641\\
1851.77198428078	0.0261249542236328\\
2055.47768467376	0.0235662460327148\\
2259.18338506674	0.0214643478393555\\
2462.88908545972	0.0197067260742188\\
2666.5947858527	0.0182151794433594\\
2870.30048624568	0.0169334411621094\\
3074.00618663866	0.0158205032348633\\
3277.71188703164	0.0148448944091797\\
3481.41758742462	0.013981819152832\\
3685.1232878176	0.013214111328125\\
3888.82898821058	0.0125265121459961\\
4092.53468860356	0.011906623840332\\
4296.24038899654	0.0113458633422852\\
4499.94608938952	0.0108346939086914\\
4703.65178978251	0.0103683471679688\\
4907.35749017549	0.00994014739990234\\
5111.06319056847	0.00954532623291016\\
5314.76889096145	0.00918102264404297\\
5518.47459135443	0.00884437561035156\\
5722.18029174741	0.00853061676025391\\
5925.88599214039	0.00823879241943359\\
6129.59169253337	0.00796604156494141\\
6333.29739292635	0.00771045684814453\\
6537.00309331933	0.00747108459472656\\
6740.70879371231	0.00724601745605469\\
6944.41449410529	0.0070343017578125\\
7148.12019449827	0.00683498382568359\\
7351.82589489125	0.00664615631103516\\
7555.53159528423	0.00646781921386719\\
7759.23729567721	0.00629806518554688\\
7962.94299607019	0.00613784790039062\\
8166.64869646317	0.00598526000976562\\
8370.35439685616	0.00583934783935547\\
8574.06009724914	0.00570106506347656\\
8777.76579764212	0.0055694580078125\\
8981.4714980351	0.00544357299804688\\
9185.17719842808	0.00532341003417969\\
9388.88289882106	0.00520801544189453\\
9592.58859921404	0.00509738922119141\\
9796.29429960702	0.00499153137207031\\
10000	0.00489044189453125\\
};
\addlegendentry{Recovery Achievable};

\addplot [color=black!70!white,dash pattern=on 10pt off 11pt,line width=2.0pt]
  table[row sep=crcr]{%
18.4206807439524	0.793833689808614\\
222.126381136933	0.128414002454846\\
425.832081529913	0.0693268325275271\\
629.537781922894	0.0474719408486928\\
833.243482315875	0.0360923984209643\\
1036.94918270886	0.0291132310266358\\
1240.65488310184	0.0243957064798149\\
1444.36058349482	0.0209938007344178\\
1648.0662838878	0.0184245294252141\\
1851.77198428078	0.0164155418157523\\
2055.47768467376	0.0148015870980349\\
2259.18338506674	0.0134765821582141\\
2462.88908545972	0.012369306449652\\
2666.5947858527	0.0114301681303532\\
2870.30048624568	0.0106235730733573\\
3074.00618663866	0.00992331171764049\\
3277.71188703164	0.00930965748678148\\
3481.41758742462	0.0087674790622444\\
3685.1232878176	0.00828497595971223\\
3888.82898821058	0.00785280982613323\\
4092.53468860356	0.00746349413314895\\
4296.24038899654	0.00711095683247409\\
4499.94608938952	0.00679022136544027\\
4703.65178978251	0.00649717027231989\\
4907.35749017549	0.00622836748062483\\
5111.06319056847	0.00598092295409614\\
5314.76889096145	0.00575238837222769\\
5518.47459135443	0.00554067584608076\\
5722.18029174741	0.00534399394631913\\
5925.88599214039	0.00516079688918225\\
6129.59169253337	0.00498974382762918\\
6333.29739292635	0.00482966597836976\\
6537.00309331933	0.00467953987980352\\
6740.70879371231	0.00453846548707704\\
6944.41449410529	0.00440564811337207\\
7148.12019449827	0.00428038345191795\\
7351.82589489125	0.00416204508253426\\
7555.53159528423	0.00405007399481772\\
7759.23729567721	0.0039439697581537\\
7962.94299607019	0.00384328304425752\\
8166.64869646317	0.0037476092665607\\
8370.35439685616	0.00365658314655382\\
8574.06009724914	0.00356987405321985\\
8777.76579764212	0.00348718199020304\\
8981.4714980351	0.00340823412806557\\
9185.17719842808	0.00333278179715957\\
9388.88289882106	0.00326059787128197\\
9592.58859921404	0.00319147448412027\\
9796.29429960702	0.00312522103012969\\
10000	0.00306166240935046\\
};
\addlegendentry{Recovery Converse};

\addplot [color=LineBlue,solid,line width=2.0pt]
  table[row sep=crcr]{%
18.4206807439524	1\\
222.126381136933	0.109429266371638\\
425.832081529913	0.0566873905504654\\
629.537781922894	0.0382538241335444\\
833.243482315875	0.0288670221884682\\
1036.94918270886	0.023179482040648\\
1240.65488310184	0.0193643376043695\\
1444.36058349482	0.0166274688188838\\
1648.0662838878	0.0145685217771221\\
1851.77198428078	0.0129633076155239\\
2055.47768467376	0.011676676156574\\
2259.18338506674	0.010622449651704\\
2462.88908545972	0.00974275261873912\\
2666.5947858527	0.00899765153334704\\
2870.30048624568	0.00835842978453814\\
3074.00618663866	0.00780398263585967\\
3277.71188703164	0.00731850418073081\\
3481.41758742462	0.00688993710920686\\
3685.1232878176	0.00650875057214631\\
3888.82898821058	0.00616754294870516\\
4092.53468860356	0.00586034439387087\\
4296.24038899654	0.00558227740110303\\
4499.94608938952	0.00532938574616936\\
4703.65178978251	0.00509843499272492\\
4907.35749017549	0.00488669282299133\\
5111.06319056847	0.00469179543752098\\
5314.76889096145	0.00451187044984106\\
5518.47459135443	0.00434522875087681\\
5722.18029174741	0.00419045171267872\\
5925.88599214039	0.00404634469871015\\
6129.59169253337	0.0039117879820835\\
6333.29739292635	0.00378591416066452\\
6537.00309331933	0.00366788528427225\\
6740.70879371231	0.00355699012502098\\
6944.41449410529	0.00345260090694147\\
7148.12019449827	0.00335418539573573\\
7351.82589489125	0.003261200397259\\
7555.53159528423	0.00317322936080516\\
7759.23729567721	0.00308989947901558\\
7962.94299607019	0.00301083303900676\\
8166.64869646317	0.00293569001398667\\
8370.35439685616	0.00286422491475071\\
8574.06009724914	0.00279615560419091\\
8777.76579764212	0.00273124566399996\\
8981.4714980351	0.0026692801243634\\
9185.17719842808	0.0026100630855732\\
9388.88289882106	0.00255341564925689\\
9592.58859921404	0.00249917411319804\\
9796.29429960702	0.00244718839137983\\
10000	0.0023973206271363\\
};
\addlegendentry{Detection Achievable};

\addplot [color=LineRed,dash pattern=on 2pt off 4pt on 8pt off 4pt,line width=2.0pt]
  table[row sep=crcr]{%
18.4206807439524	0.0261611938476562\\
222.126381136933	0.00242900848388672\\
425.832081529913	0.00127410888671875\\
629.537781922894	0.000863075256347656\\
833.243482315875	0.000652313232421875\\
1036.94918270886	0.000524520874023438\\
1240.65488310184	0.000438690185546875\\
1444.36058349482	0.000376701354980469\\
1648.0662838878	0.000329971313476562\\
1851.77198428078	0.000293731689453125\\
2055.47768467376	0.000265121459960938\\
2259.18338506674	0.000241279602050781\\
2462.88908545972	0.00022125244140625\\
2666.5947858527	0.000204086303710938\\
2870.30048624568	0.000189781188964844\\
3074.00618663866	0.000177383422851562\\
3277.71188703164	0.000165939331054688\\
3481.41758742462	0.000156402587890625\\
3685.1232878176	0.000147819519042969\\
3888.82898821058	0.000140190124511719\\
4092.53468860356	0.000132560729980469\\
4296.24038899654	0.000126838684082031\\
4499.94608938952	0.000121116638183594\\
4703.65178978251	0.000115394592285156\\
4907.35749017549	0.000110626220703125\\
5111.06319056847	0.000105857849121094\\
5314.76889096145	0.000102043151855469\\
5518.47459135443	9.82284545898438e-05\\
5722.18029174741	9.5367431640625e-05\\
5925.88599214039	9.1552734375e-05\\
6129.59169253337	8.86917114257812e-05\\
6333.29739292635	8.58306884765625e-05\\
6537.00309331933	8.29696655273438e-05\\
6740.70879371231	8.0108642578125e-05\\
6944.41449410529	7.82012939453125e-05\\
7148.12019449827	7.62939453125e-05\\
7351.82589489125	7.34329223632812e-05\\
7555.53159528423	7.15255737304688e-05\\
7759.23729567721	6.96182250976562e-05\\
7962.94299607019	6.77108764648438e-05\\
8166.64869646317	6.67572021484375e-05\\
8370.35439685616	6.4849853515625e-05\\
8574.06009724914	6.29425048828125e-05\\
8777.76579764212	6.19888305664062e-05\\
8981.4714980351	6.00814819335938e-05\\
9185.17719842808	5.91278076171875e-05\\
9388.88289882106	5.7220458984375e-05\\
9592.58859921404	5.62667846679688e-05\\
9796.29429960702	5.53131103515625e-05\\
10000	5.43594360351562e-05\\
};
\addlegendentry{Detection Converse};

\end{axis}
\end{tikzpicture}%

%% file: riskvsntikzfig.tex
%
%
%
\definecolor{mycolor1}{rgb}{0.00000,1.00000,1.00000}%
\begin{tikzpicture}

\begin{axis}[%
width=2.6in,
height=2.53in,
scale only axis,
separate axis lines,
every outer x axis line/.append style={white!15!black},
every x tick label/.append style={font=\color{white!15!black}},
xmin=0,
xmax=20000,
x tick label style={/pgf/number format/.cd,%
          scaled x ticks = false,
          set thousands separator={},
          fixed},
xlabel={$n$, Number of Users},
every outer y axis line/.append style={white!15!black},
every y tick label/.append style={font=\color{white!15!black}},
ymode=log,
ymin=1e-05,
ymax=1,
yminorticks=true,
ylabel={$\rho{}^\text{2}$, Squared Correlation Coefficient},
title={Number of Features $d = 1000$, Risk $R = 0.1$},
legend style={at={(0.352321428571429,0.013319047619048)},anchor=south west,draw=white!15!black,fill=white,legend cell align=left}
]
\addplot [color=LineGreen,dotted,line width=2.0pt]
  table[row sep=crcr]{%
100	0.0301361083984375\\
506.122448979592	0.0364065170288086\\
912.244897959184	0.0386743545532227\\
1318.36734693878	0.0400896072387695\\
1724.48979591837	0.0411195755004883\\
2130.61224489796	0.04193115234375\\
2536.73469387755	0.0425987243652344\\
2942.85714285714	0.0431680679321289\\
3348.97959183673	0.0436620712280273\\
3755.10204081633	0.0440998077392578\\
4161.22448979592	0.0444927215576172\\
4567.34693877551	0.0448484420776367\\
4973.4693877551	0.0451736450195312\\
5379.59183673469	0.0454740524291992\\
5785.71428571429	0.0457515716552734\\
6191.83673469388	0.0460100173950195\\
6597.95918367347	0.0462532043457031\\
7004.08163265306	0.0464811325073242\\
7410.20408163265	0.0466957092285156\\
7816.32653061224	0.0468988418579102\\
8222.44897959184	0.0470924377441406\\
8628.57142857143	0.0472755432128906\\
9034.69387755102	0.0474510192871094\\
9440.81632653061	0.0476188659667969\\
9846.9387755102	0.0477790832519531\\
10253.0612244898	0.0479326248168945\\
10659.1836734694	0.0480813980102539\\
11065.306122449	0.0482234954833984\\
11471.4285714286	0.0483608245849609\\
11877.5510204082	0.0484933853149414\\
12283.6734693878	0.0486211776733398\\
12689.7959183673	0.0487451553344727\\
13095.9183673469	0.0488643646240234\\
13502.0408163265	0.048980712890625\\
13908.1632653061	0.0490932464599609\\
14314.2857142857	0.0492029190063477\\
14720.4081632653	0.0493097305297852\\
15126.5306122449	0.049412727355957\\
15532.6530612245	0.0495138168334961\\
15938.7755102041	0.0496120452880859\\
16344.8979591837	0.0497074127197266\\
16751.0204081633	0.0498008728027344\\
17157.1428571429	0.049891471862793\\
17563.2653061224	0.0499801635742188\\
17969.387755102	0.0500679016113281\\
18375.5102040816	0.0501527786254883\\
18781.6326530612	0.0502357482910156\\
19187.7551020408	0.0503168106079102\\
19593.8775510204	0.0503959655761719\\
20000	0.0504741668701172\\
};
\addlegendentry{Recovery Achievable};

\addplot [color=black!70!white,dash pattern=on 10pt off 11pt,line width=2.0pt]
  table[row sep=crcr]{%
100	0.0121592127579635\\
506.122448979592	0.0185396805643344\\
912.244897959184	0.0208474827818615\\
1318.36734693878	0.022287238464867\\
1724.48979591837	0.0233358352836389\\
2130.61224489796	0.0241608337857673\\
2536.73469387755	0.0248409284452382\\
2942.85714285714	0.0254194351745436\\
3348.97959183673	0.0259227564347037\\
3755.10204081633	0.026368179625859\\
4161.22448979592	0.026767642385667\\
4567.34693877551	0.0271297345788589\\
4973.4693877551	0.0274608435109627\\
5379.59183673469	0.0277658478512829\\
5785.71428571429	0.0280485581520549\\
6191.83673469388	0.0283120074808934\\
6597.95918367347	0.0285586494803736\\
7004.08163265306	0.0287904971263797\\
7410.20408163265	0.0290092222914557\\
7816.32653061224	0.0292162286902196\\
8222.44897959184	0.0294127063138446\\
8628.57142857143	0.0295996727187562\\
9034.69387755102	0.0297780048039096\\
9440.81632653061	0.0299484635903079\\
9846.9387755102	0.0301117137740263\\
10253.0612244898	0.0302683393220927\\
10659.1836734694	0.0304188560348736\\
11065.306122449	0.0305637217565038\\
11471.4285714286	0.0307033447426961\\
11877.5510204082	0.0308380905710588\\
12283.6734693878	0.0309682878883023\\
12689.7959183673	0.0310942332215998\\
13095.9183673469	0.0312161950311974\\
13502.0408163265	0.0313344171434654\\
13908.1632653061	0.0314491216746734\\
14314.2857142857	0.0315605115335342\\
14720.4081632653	0.0316687725732965\\
15126.5306122449	0.0317740754506717\\
15532.6530612245	0.031876577238244\\
15938.7755102041	0.0319764228285785\\
16344.8979591837	0.0320737461615022\\
16751.0204081633	0.0321686713006238\\
17157.1428571429	0.0322613133807819\\
17563.2653061224	0.0323517794445529\\
17969.387755102	0.0324401691830543\\
18375.5102040816	0.0325265755938827\\
18781.6326530612	0.0326110855670694\\
19187.7551020408	0.0326937804082981\\
19593.8775510204	0.032774736307278\\
20000	0.0328540247580271\\
};
\addlegendentry{Recovery Converse};

\addplot [color=LineBlue,solid,line width=2.0pt]
  table[row sep=crcr]{%
100	0.0240385162152998\\
506.122448979592	0.0240385162152998\\
912.244897959184	0.0240385162152998\\
1318.36734693878	0.0240385162152998\\
1724.48979591837	0.0240385162152998\\
2130.61224489796	0.0240385162152998\\
2536.73469387755	0.0240385162152998\\
2942.85714285714	0.0240385162152998\\
3348.97959183673	0.0240385162152998\\
3755.10204081633	0.0240385162152998\\
4161.22448979592	0.0240385162152998\\
4567.34693877551	0.0240385162152998\\
4973.4693877551	0.0240385162152998\\
5379.59183673469	0.0240385162152998\\
5785.71428571429	0.0240385162152998\\
6191.83673469388	0.0240385162152998\\
6597.95918367347	0.0240385162152998\\
7004.08163265306	0.0240385162152998\\
7410.20408163265	0.0240385162152998\\
7816.32653061224	0.0240385162152998\\
8222.44897959184	0.0240385162152998\\
8628.57142857143	0.0240385162152998\\
9034.69387755102	0.0240385162152998\\
9440.81632653061	0.0240385162152998\\
9846.9387755102	0.0240385162152998\\
10253.0612244898	0.0240385162152998\\
10659.1836734694	0.0240385162152998\\
11065.306122449	0.0240385162152998\\
11471.4285714286	0.0240385162152998\\
11877.5510204082	0.0240385162152998\\
12283.6734693878	0.0240385162152998\\
12689.7959183673	0.0240385162152998\\
13095.9183673469	0.0240385162152998\\
13502.0408163265	0.0240385162152998\\
13908.1632653061	0.0240385162152998\\
14314.2857142857	0.0240385162152998\\
14720.4081632653	0.0240385162152998\\
15126.5306122449	0.0240385162152998\\
15532.6530612245	0.0240385162152998\\
15938.7755102041	0.0240385162152998\\
16344.8979591837	0.0240385162152998\\
16751.0204081633	0.0240385162152998\\
17157.1428571429	0.0240385162152998\\
17563.2653061224	0.0240385162152998\\
17969.387755102	0.0240385162152998\\
18375.5102040816	0.0240385162152998\\
18781.6326530612	0.0240385162152998\\
19187.7551020408	0.0240385162152998\\
19593.8775510204	0.0240385162152998\\
20000	0.0240385162152998\\
};
\addlegendentry{Detection Achievable};

\addplot [color=LineRed,dash pattern=on 2pt off 4pt on 8pt off 4pt,line width=2.0pt]
  table[row sep=crcr]{%
100	0.00541877746582031\\
506.122448979592	0.002410888671875\\
912.244897959184	0.00179862976074219\\
1318.36734693878	0.00149440765380859\\
1724.48979591837	0.00130939483642578\\
2130.61224489796	0.00117683410644531\\
2536.73469387755	0.00107955932617188\\
2942.85714285714	0.00100135803222656\\
3348.97959183673	0.000939369201660156\\
3755.10204081633	0.000887870788574219\\
4161.22448979592	0.000843048095703125\\
4567.34693877551	0.000804901123046875\\
4973.4693877551	0.000771522521972656\\
5379.59183673469	0.000741004943847656\\
5785.71428571429	0.000715255737304688\\
6191.83673469388	0.000691413879394531\\
6597.95918367347	0.000669479370117188\\
7004.08163265306	0.000650405883789062\\
7410.20408163265	0.000631332397460938\\
7816.32653061224	0.000615119934082031\\
8222.44897959184	0.000599861145019531\\
8628.57142857143	0.000585556030273438\\
9034.69387755102	0.00057220458984375\\
9440.81632653061	0.000559806823730469\\
9846.9387755102	0.000547409057617188\\
10253.0612244898	0.000536918640136719\\
10659.1836734694	0.00052642822265625\\
11065.306122449	0.000516891479492188\\
11471.4285714286	0.000507354736328125\\
11877.5510204082	0.000498771667480469\\
12283.6734693878	0.000491142272949219\\
12689.7959183673	0.000482559204101562\\
13095.9183673469	0.000474929809570312\\
13502.0408163265	0.000468254089355469\\
13908.1632653061	0.000460624694824219\\
14314.2857142857	0.000454902648925781\\
14720.4081632653	0.000448226928710938\\
15126.5306122449	0.0004425048828125\\
15532.6530612245	0.000435829162597656\\
15938.7755102041	0.000431060791015625\\
16344.8979591837	0.000425338745117188\\
16751.0204081633	0.000420570373535156\\
17157.1428571429	0.000414848327636719\\
17563.2653061224	0.000410079956054688\\
17969.387755102	0.000405311584472656\\
18375.5102040816	0.000401496887207031\\
18781.6326530612	0.000396728515625\\
19187.7551020408	0.000392913818359375\\
19593.8775510204	0.000388145446777344\\
20000	0.000384330749511719\\
};
\addlegendentry{Detection Converse};

\end{axis}
\end{tikzpicture}%

%% file: 4-Ach.tex
\section{Achievability}\label{sec: ach-inner}
Our achievability bound follows directly from the Chernoff bound, since the moment generating function (MGF) for our proposed sum-of-inner-products statistic $T = \sign{\rho} \sum_{i,j} X_i^\mathsf{T}Y^{\vphantom{\mathsf{T}}}_j$ can be calculated explicitly under $H_0$ and $H_1$. 

\begin{lemma}\label{lemma: mgf-T} For $-\frac{1}{n(1-\abs{\rho})} < \lambda < \frac{1}{n(1+\abs{\rho})}$, the MGF under $H_1$ is
	\begin{align*}
		\expe_{1|\sigma}[e^{\lambda T}] &= \left(1-2n\lambda|\rho| - n^2\lambda^2 (1-\rho^2)\right)^{-d/2} \ . \end{align*} For $\abs{\lambda }< \frac{1}{n}$, the MGF under $H_0$ is 
	\begin{align*}
		\expe_{0}[e^{\lambda T}] &= \big(\expe_{1|\sigma}[e^{\lambda T}]\big) \big|_{\rho=0}  = \big(1 - n^2\lambda^2 \big)^{-d/2} \ .
	\end{align*}
\end{lemma} 
\begin{proof} Observe that under $\mathbb{P}_{1|\sigma}$, $X_i \stackrel{d}{=} \rho Y_{\sigma_i}+\sqrt{1-\rho^2}Z_i$ with $Z_i \iid \mathcal{N}(0_d,\id_d)$ and independent of $Y_{\sigma_i}$ for any $i \in [n]$.  For some $\lambda$, we get
	\begin{align*}
		\expe_{1|\sigma}[e^{\lambda T} ]
		&= \expe_{1|\sigma}\Big[ \exp \Big(\lambda \,\sign{\rho}\sum_{i,j}^n X_i^\mathsf{T} Y_j^{\vphantom{T}}\Big)\Big] \\ 
		&= \expe_{Y^n, Z^n\,\iid \mathcal{N}(0_d,\id_d)} \Big[\exp \Big(\lambda \sum_{i,j}^n (|\rho| Y_{\sigma_i}+\sign{\rho}\sqrt{1-\rho^2}Z_i)^\mathsf{T} Y_j^{\vphantom{\mathsf{T}}}\Big)\Big] \\ 
		& \stackrel{\text{(i)}}{=} \expe \bigg[ \exp \Big( \lambda |\rho|\sum_{i,j} Y_j^\mathsf{T} Y_{\sigma_i}^{\vphantom{T}}\Big) \expe \Big[\exp \Big( \lambda\,\sign{\rho}\sqrt{1-\rho^2}\,\sum_{i} Z_i^\mathsf{T}\Big(\sum_j Y_j^{\vphantom{\mathsf{T}}}\Big) \Big) \mid Y^n \Big] \bigg] \\
		&	\stackrel{\text{(ii)}}{=} \expe \bigg[ \exp \Big( \lambda |\rho|\sum_{i,j} Y_j^\mathsf{T} Y_{\sigma_i}^{\vphantom{T}}\Big)  \prod_{i} \expe \Big[\exp \Big( \lambda\,\sign{\rho}\sqrt{1-\rho^2}\, 	Z_i^\mathsf{T}\Big(\sum_j Y_j^{\vphantom{\mathsf{T}}}\Big) \Big) \mid Y^n \Big] \bigg] \\
		& \stackrel{\text{(iii)}}{=} \expe\bigg[ \exp\bigg( \lambda |\rho|\sum_{i,j}  Y_{j}^\mathsf{T} Y_{\sigma_i}^{\vphantom{T}}+n \frac{\lambda^2 (1-\rho^2)}{2} \bigg|\bigg|{\sum_j Y_j}\bigg|\bigg|^2 \bigg) \bigg]
	\end{align*} 
	where (i) follows by the law of total expectation, (ii) the conditional independence of $Z^n$ given $Y^n$ and (iii) the moment generating function of a normal variable. Due to the bijectivity of a permutation, we  observe that $\sum_{i,j} Y_j^\mathsf{T} Y_{\sigma_i}^{\vphantom{T}}=\norm{\sum_j Y_j}^2 $. To simplify the notation, let $Y\stackrel{d}{=}\sum_j Y_j$ so that  $\frac{Y}{\sqrt{n}} \sim \mathcal{N}(0,\id_d)$. Plugging back into the last equation and using the moment generating function of the $\chi^2_d$ distribution, we get
	\begin{align*}
		\expe_{1|\sigma}[e^{\lambda T} ] &= \expe \bigg[\exp \bigg(  \norm{\frac{Y}{\sqrt{n}}}^2 \Big( n\lambda |\rho| + n^2\, \frac{\lambda^2 (1-\rho^2)}{2} \Big) \bigg)\bigg]  \\ 
		& = \left(1-2n\lambda|\rho| - n^2\lambda^2 (1-\rho^2)\right)^{-d/2} 
	\end{align*}
	for $1-2n\lambda|\rho| - n^2\lambda^2 (1-\rho^2) >0$. By completing the square, we can rewrite the constraint as
	\begin{align*}
		1-2n\lambda|\rho| - n^2\lambda^2 (1-\rho^2) >0  
		&\Leftrightarrow \left( \lambda + \frac{\abs{\rho}}{n(1-\rho^2)} \right) ^2 < \frac{1}{n^2(1-\rho^2)} + \frac{\rho^2}{n^2(1-\rho^2)^2} \\
		&\Leftrightarrow  \frac{-1}{n(1-\abs{\rho})} < \lambda < \frac{1}{n(1+\abs{\rho})} \ .
	\end{align*} 
	For $\expe_0[e^{\lambda T}]$, we simply note that the null model is equal to the alternate model evaluated at zero correlation, $\P_0 = (\P_{1 | \sigma})|_{\rho = 0}$, and thus $\expe_0[e^{\lambda T}] = \big(\expe_{1|\sigma}[e^{\lambda T}]\big)\big|_{\rho = 0}$.
\end{proof}

We can now apply the Chernoff bound to establish our achievability result.

\begin{proof}[Proof of Theorem~\ref{thm:detectionachievable}] From the definition of the threshold test~\eqref{eq: siptest} with $t \in (\expe_{0}\,T,\expe_{1| \sigma}\, T) = (0,|\rho|nd)$, we have that
	\begin{align*} 
		R(\phi_T) &= \min_t \Big(\P_0\{T \geq t\}  + \max_{\sigma \in S_n}\P_{1|\sigma} \{T < t\}  \Big)  \ . 
\intertext{Next, we apply the Chernoff bound by using the moment generating functions in Lemma~\ref{lemma: mgf-T}. For the false alarm probability, we have}
		\P_0\{\phi_{T,t}=1\} = \P_0\{T \geq t\} &\leq\;\, \min_{\lambda > 0}\;\, \exp\Big(-\lambda t + \ln\big( \expe_0[e^{\lambda T}]\big)\Big) \\ &= \min_{0 <  \lambda <\frac{1}{n}} \exp \Big( -\lambda t +\frac{d}{2} \ln\frac{1}{1-n^2 \lambda^2} \Big) \\ &=\exp\Big(-\frac{d}{2} g_{\mathrm{FA}}(\gamma) \Big) 
		\intertext{	where the last step follows from plugging in the optimal value $\lambda^*_{\text{FA}} = -\frac{d}{2t} + \sqrt{\frac{1}{n^2} + (\frac{d}{2t})^2}$ (See Lemma~\ref{lemma: minimize-g}), defining $\gamma = (\frac{2t}{dn})^2 \in (0,4\rho^2)$, and simplifying. Similarly, for the missed detection probability, we have}
		\P_{1|\sigma} \{\phi_{T,t}=0\}  = \P_{1|\sigma} \{T < t\} 
		&\leq\;\; \min_{\lambda > 0 }\;\; \exp\Big(\lambda t + \ln\big( \expe_{1| \sigma} [e^{\lambda T}]\big)\Big)\\
		&=   \min_{\; 0 < \lambda < \frac{1}{n(1-\abs{\rho})}} \exp\Big( \lambda t - \frac{d}{2} \ln\big(1+2n\lambda|\rho| - n^2\lambda^2 (1-\rho^2)\, \big) \Big)  \\
		&=  \min_{\; 0 < \lambda < \frac{1}{n(1-\abs{\rho})}} \exp \left(  \lambda t + \frac{d}{2} \ln \frac{1-\rho^2}{1- \big(n(1-\rho^2)\lambda -|\rho|\big)^2} \right) \\
		&=\exp\Big(-\frac{d}{2} g_{\mathrm{MD}}(\gamma) \Big) 
	\end{align*}
	where the last step follows from plugging in the optimal value $\lambda^*_{\text{MD}} = \frac{|\rho|}{n(1-\rho^2)} + \frac{d}{2t} - \sqrt{(\frac{1}{n(1-\rho^2)})^2 + (\frac{d}{2t})^2}$, defining $\gamma = (\frac{2t}{dn})^2 \in (0,4\rho^2)$ (See Lemma~\ref{lemma: minimize-g}), and simplifying. 
	
	First, we observe that by using the inequality $\ln x \leq x-1$ and the concavity of the square root function, we get
	\begin{equation}\label{eq:fa-exponent}
	g_{\text{FA}}(\gamma) \geq \frac{\sqrt{2}-1}{2}\gamma \ .
	\end{equation}
	
	Now, the aim becomes to pick $\gamma \in (0, 4\rho^2)$  (and consequently to pick $t$) such that the  difference $g_{\mathrm{FA}} -g_{\mathrm{MD}}= 0$, i.e., to pick the threshold $t$ equates both type of errors. Since it is challenging to find the exact value, we put upper bound on it by setting $\gamma = \rho^2$. Consequently, in Lemma~\ref{lemma: ach-f-func}, we show that 
	\begin{equation*}
g_{\text{MD}}(\rho^2) \geq g_{\text{FA}}(\rho^2) -(\sqrt{2}-1)^2\rho^2
\end{equation*}
	and combining this  with~\eqref{eq:fa-exponent}, we get
		\begin{equation*}
		g_{\text{MD}}(\rho^2) \geq \frac{\sqrt{2}-1}{2}\rho^2 -(\sqrt{2}-1)^2\rho^2 \geq \frac{\rho^2}{30} \ .
	\end{equation*}
	Finally, for this chosen $t = \sqrt{\gamma}\frac{dn}{2}= |\rho|\frac{dn}{2}$, the concise (but loose) upper bound $2 \exp\big(-\frac{d \rho^2}{60}\big)$ follows. 
\end{proof} 

%% file: 5-Converse.tex
\section{Converse}\label{sec: converse}
We begin with the well-known fact that the minimax risk is lower bounded by the Bayes risk. For a prior distribution $\pi$ over the set of permutations $S_n$, we define the corresponding Bayes risk as
$$R_\pi^*\triangleq  \inf_{\phi}\, \mathbb{P}_0\big[ \phi(X^n,Y^n) = 1 \big]+ \expe_{\sigma \sim \pi}\, \mathbb{P}_{1|\sigma} \big[ \phi(X^n,Y^n) = 0 \big] \ .$$

Then, $R^* \geq  R_\pi^*$ for any prior $\pi$. The Bayes risk stems from a binary hypothesis test between simple hypotheses, $\P_0$ for the null and the mixture distribution $\P_1=\mathbb{E}_{\sigma \sim \pi }\, \P_{1|\sigma}$ for the alternate. Since $\P_0$ and $\P_1$ are mutually absolutely continuous, we define their likelihood ratio as $L \triangleq \P_1 /\P_0$. It follows that 
\begin{equation} \label{eq: risklower}
	1\!-\! R^* \!\leq  1\!-\! R_\pi^* \! = d_{\mathrm{TV}}(\P_0,\P_1) \! \overset{\text{(i)}}{=\vphantom{\leq}} \expe_{0}\abs{L \!-\!1} \overset{\text{(ii)}}{\leq} \! \sqrt{\expe_{0} L^2 \!-\!1} 
\end{equation}
where (i) uses the fact that the total variation distance satisfies $d_{\mathrm{TV}}(P,Q) = \half \expe_Q \big\vert{ \frac{P}{Q}-1}\big\vert $ and (ii) is due to Cauchy-Schwarz inequality.

This second-order moment method  (i.e., upper bounding $\expe_{0}\, L^2$)  is a common approach to find the lower bound for a hypothesis testing problems (see, e.g.,
\cite{wu2021statistical,jiaming2021testinggraphs, arias2014commdetect}). The composite law under $H_1$ makes it challenging to compute $\mathbb{E}_0\, L^2 = \mathbb{E}_0 (\, \mathbb{E}_{\sigma \sim\pi} \,\mathbb{P}_{1|\sigma} \, /\, \mathbb{P}_0)^2$. The lemma below gives an alternative expression by using Fubini's theorem to exchange the order of integration. To the best of our knowledge, this approach was introduced by Ingster and Suslina in the context of a  Gaussian location model (see \cite[Equation (3.66)]{ingster2003nonparametric}).

\begin{lemma}[Ingster-Suslina method]\label{lemma: ingster-suslina}
	Let $\mathbb{P}_{1|\sigma}$ and $\mathbb{P}_0$ be distributions with $\sigma $  from a parameter space. Given a prior on the space  of $\sigma$, define the  mixture distribution as $\mathbb{P}_1 = \mathbb{E}_{\sigma \sim \pi}\,  \mathbb{P}_{1|\sigma}$. 	Then, with $\sigma, \tilde{\sigma} \iid \pi$ and $G(\sigma, \tilde{\sigma}) = \displaystyle \int \frac{\mathbb{P}_{1|\sigma} \mathbb{P}_{1|\tilde{\sigma}}}{\mathbb{P}_0} $, we have $$\mathbb{E}_0\, L^2 = \mathbb{E}_0 \left(\frac{\mathbb{E}_{\sigma \sim \pi }\, \mathbb{P}_{1|\sigma} }{\mathbb{P}_0}\right)^2 = \expe_{\sigma, \tilde{\sigma} {\sim} \pi}\,{G(\sigma, \tilde{\sigma})} \ .$$ 
\end{lemma}

This leads to our first converse result.
\begin{lemma}\label{lemma: converse-unconditional}
	$$R^*\geq 1 - \sqrt{(1-\rho^2)^{-dn}-1}$$
	Consequently, if $\rho^2 = o\big(\frac{1}{dn}\big)$, then $R^* \to 1$ as $n\to \infty$ or $d\to \infty$.
\end{lemma}
The proof uses a version of \cite[Proposition 1]{jiaming2021testinggraphs}, Lemma \ref{lemma: ingster-suslina} with $\pi = \mathrm{Unif}(S_n)$ and Equation \eqref{eq: risklower}. See Appendix~\ref{app: unconditional-converse} for details.

\subsection{Truncation}\label{sec: truncated}

To sharpen the converse from $o\big(\frac{1}{dn}\big)$ to $o\big(\frac{1}{d\sqrt{n}}\big)$, we truncate the likelihood ratio by conditioning on a carefully-chosen event that occurs with high probability under $H_1$, following the approach of \cite{arias2014commdetect, jiaming2021testinggraphs}.
For an event $\Gamma_\sigma$, define the truncated likelihood ratio as
$$\tilde{L} = \expe_{\sigma\sim\pi}\, {L_\sigma \mathbbm{1}_{\Gamma_\sigma}} \text{ where } L_\sigma = \frac{\mathbb{P}_{1|\sigma}}{\mathbb{P}_0} \ .$$
By the triangle and Cauchy-Schwarz inequalities combined with the fact that $\tilde{L} \leq L = \expe_{\sigma\sim\pi}\, L_\sigma $, we get
\begin{align}\label{eq: risk-truncated}
	1- R^* \leq 1-R^*_\pi & = \mathbb{E}_0 \,|L-1| \leq  \mathbb{E}_0 \,|\tilde{L}-1| + \mathbb{E}_0 (L-\tilde{L} ) \nonumber \\
	&= \sqrt{\mathbb{E}_0 \,\tilde{L}^2  -1+2\big(1- \mathbb{E}_0\, \tilde{L} \big) } + \big(1- \mathbb{E}_0\, \tilde{L} \big)  \ .
\end{align}
Thus, if we simultaneously lower bound $\mathbb{E}_0\, \tilde{L}$ and upper bound on $\mathbb{E}_0\, \tilde{L}^2$, we can obtain a converse bound.
Observe that by Fubini's theorem, the truncated first-order moment becomes $$\mathbb{E}_0 \,\tilde{L} = \mathbb{E}_0\, \mathbb{E}_{\sigma\sim\pi} L_\sigma \mathbbm{1}_{\Gamma_\sigma}  
= \mathbb{E}_{\sigma\sim\pi} \, \mathbb{E}_0\, L_\sigma \mathbbm{1}_{\Gamma_\sigma}   
= \mathbb{E}_{\sigma\sim\pi} \, \mathbb{P}_{1|\sigma}(\Gamma_\sigma) \ .$$ 
Setting $\pi = \mathrm{Unif}(S_n)$, we will select an event $\Gamma_\sigma$ that occurs with high probability under $H_1$ for any $\sigma \in S_n$, which in turns shows that $\mathbb{E}_0\, \tilde{L}\to 1$. Moreover, Lemma \ref{lemma: ingster-suslina} can be used to compute $\mathbb{E}_0\, \tilde{L}^2$ as before.

For a permutation $\sigma \in S_n$, let $F_\sigma = \{i \in [n]: i=\sigma_i\}$ be the set of corresponding fixed points and let $N_1^\sigma = |F_\sigma|$. For any $\sigma$,  define the truncation event $\Gamma_\sigma$ as
\begin{equation}\label{eq: truncation}
\Gamma_\sigma \triangleq \bigcap_{k=k^*}^{N_1^\sigma}\,\bigcap_{\substack{T \subseteq F_\sigma \\ |T|=k}} \left\{\,\sum_{i\in T} X_i^\mathsf{T} X^{\vphantom{\mathsf{T}}}_i > w_k ,\,\sum_{i\in T} Y_i^\mathsf{T} Y^{\vphantom{\mathsf{T}}}_i > w_k,\, 
 \sign{\rho} \sum_{i\in T} X_i^\mathsf{T}Y_{\sigma_i}^{\vphantom{\mathsf{T}}} < v_k\right\}
\end{equation}
where $k^*\in [n],\, w_k >0$ and  $v_k $ for $k=k^*,\ldots,N_1^\sigma$ will be chosen later. 

In the following, we will bound the first and second  moments of by the truncated likelihood ratio  $\tilde{L} = \expe_{\sigma\sim\mathrm{Unif}(S_n)}\, \frac{\mathbb{P}_{1|\sigma}}{\mathbb{P}_0} \mathbbm{1}_{\Gamma_\sigma}$ where $\P_0$ and $\P_{1|\sigma}$ are the likelihoods of the binary hypothesis testing problem~\eqref{eq: problem} with prior distribution $\mathrm{Unif}(S_n)$ on $S_n$.

\begin{lemma}[Truncated first-order moment]\label{lemma: truncated-first}
	For $d\geq 4\ln \frac{en}{\;k^*}$, let
	\begin{align*}
		&w_k = dk-2\sqrt{d}k\,r_k \quad \text{and} \quad v_k = \abs{\rho} dk + 4\abs{\rho} \sqrt{d}k \,s_k
	\end{align*}
	for  all $k=k^*,\ldots, n$. Then, we have
	$$1-\mathbb{E}_0\, \tilde{L} \leq \frac{4e^{-k^*\min\{\psi_1,\psi_2\}}}{1-e^{-\min\{\psi_1,\psi_2\}}} $$
	where
	$$ \psi_1 = \min_{k=k^*,\ldots, n}\, - \ln \frac{e n}{k} +r_k^2 \quad\text{and}$$
	$$\psi_2 = \min_{k=k^*,\ldots, n} \frac{|\rho| \sqrt{d} \,s_k }{4}\min \left\{\frac{1}{|\rho|},\, \frac{2}{\sqrt{1-\rho^2}},\, \frac{s_k }{|\rho|\sqrt{d}},\, \frac{4|\rho| s_k }{(1-\rho^2)\sqrt{d}}\right\} - \ln \frac{e n}{k} \ . $$
	Consequently, if $k^*=\omega(1)$ and $\min\{\psi_1,\psi_2\}>0$, then $\expe_0 \,\tilde{L} \to 1$ as $n\to \infty$	with $\expe_0 \,\tilde{L}  \geq 1-\mathcal{O}(1)$.
\end{lemma}
\begin{proof}[Proof Sketch]
	As discussed before, it is enough to upper bound $ \mathbb{P}_{1|\sigma}(\Gamma_\sigma^C)$ for any $\sigma$ where $\Gamma_\sigma^C$ is the complement of $\Gamma_\sigma$. Upon applying the union bound, we get
	$$	\mathbb{P}_{1|\sigma}(\,\Gamma_\sigma^C\,)\leq  \sum_{k=k^*}^{N_1^\sigma} \, \sum_{\substack{T \subseteq F_\sigma \\ |T|=k}} 2 \mathbb{P}_{1|\sigma}\Big(\sum_{i \in T} \norm{X_i}^2 \leq w_k\Big)+  \mathbb{P}_{1|\sigma}\Big( \sign{\rho}\sum_{i \in T} X_i^\mathsf{T}Y_{\sigma_i}^{\vphantom{\mathsf{T}}} \geq v_k\Big) \ . $$
	Then, the proof follows by applying two concentration inequalities: Laurent-Massart lemma  (Lemma~\ref{laurent-massart})  and Gaussian chaos  (Lemma~\ref{gaussian-chaos}), applied to the first and second terms, respectively. See Appendix~\ref{app: truncated-first} for details.
\end{proof}

\begin{lemma}[Truncated second-order moment]\label{lemma: truncated-second}
	If the conditions of Lemma ~\ref{lemma: truncated-first} hold  for $\expe_0 \,\tilde{L} \to 1$, then we have
	$$\mathbb{E}_0 \,\tilde{L}^2  \leq \exp\left( \frac{dn}{2} \left(\frac{\rho^2}{1-\rho^2}\right)^2+ dk^*\frac{\rho^2}{1-\rho^2}\right)  + \frac{e^{-k^*\psi}}{1-e^{-\psi}}$$
	where
		$$	\psi = \min_{k=k^*,\ldots, n}  -\frac{dn}{2k} \frac{\rho^4}{1-\rho^4} - d\frac{\rho^2}{1-\rho^2} +\frac{2\rho^2}{1-\rho^2} \Big(\frac{w_k }{k}- \frac{v_k}{k|\rho|}  \Big) + \ln \frac{k}{e} \ .$$
	Consequently, by choosing $k^*=13\sqrt{n}$, if $n\geq e^2$  and
	$$\rho^2  =o\left(\frac{1}{d\sqrt{n}}\right) ,$$ 
	then  $\expe_0 \,\tilde{L}^2 \to 1$ as $n\to \infty$ and $d\to\infty$ with $\expe_0 \,\tilde{L} ^2 \leq  1+\mathcal{O}(1)$.
\end{lemma}
\begin{proof}[Proof Sketch]
	The proof follows by examining the behavior of $\expe_0 \,\tilde{L} ^2$ for small and large number of fixed points of a permutation separately as 
	$$\expe_0 \,\tilde{L} ^2 =\underbrace{ \expe_0 \,\tilde{L} ^2  \mathbbm{1}\{ N_1^\sigma \leq k^*\}}_{\triangleq \,(\text{I})} + \underbrace{\expe_0 \,\tilde{L} ^2 \mathbbm{1}\{ N_1^\sigma > k^*\}}_{\triangleq \,(\text{II})}  \ .$$
	Then, the aim is converted into finding two functions of $k^*$, call them $f_1(k^*)$ and $f_2(k^*)$, such that 
	\begin{align*}
		\text{for } \rho^2< f_1(k^*),& \quad (\text{I})= 1+\mathcal{O}(1) \text{ and} \\
		\text{for }  \rho^2< f_2(k^*), & \quad (\text{II}) = \mathcal{O}(1)
	\end{align*}
	both hold. Combining them, we find the desired condition:
	$$ \text{For }\rho^2 < \min\{f_1(k^*),\, f_2(k^*)\}, \; \expe_0 \,\tilde{L} ^2= \text{(I)}+\text{(II)} = 1+\mathcal{O}(1) \ . $$
	In order to find a converse region of  $\rho^2$ as large as possible -- as tight as possible -- we need to choose $k^*$ as the maximizer of $\min\{f_1(k^*),\, f_2(k^*)\}$. See Appendix~\ref{app: truncated-second} for details.
\end{proof}

We can now establish the tighter converse result. 
\begin{proof}[Proof of Theorem~\ref{thm:detectionconverse}]
	Apply Equation~\eqref{eq: risk-truncated} along with Lemma~\ref{lemma: truncated-first} and Lemma~\ref{lemma: truncated-second} by setting the prior distribution as $\pi=\mathrm{Unif}(S_n)$, if $\rho^2 = o\big(\frac{1}{d\sqrt{n}}\big)$ and $d=\Omega\big(\ln \frac{n}{\,k^*}\big) = \Omega(\ln n)$ for $k^*= 13\sqrt{n}$, then one gets 
	$R^*\to 1$ as $d\to\infty$ and $n\to \infty$.
\end{proof}

%% file: app-Converse.tex
\section{Omitted Proofs from Section~\ref{sec: converse}}\label{app: converse}
Throughout the section, when we write $\mathcal{N}(X)$ it will be understood as the density function of $X\sim \mathcal{N}(0,\id_d)$ and  $\mathcal{N}_{\rho}(Y|X)$ will be understood as the density function of $Y|X \sim \mathcal{N}(\rho X , (1-\rho^2)\id_d)$. Observe that  $\Sigma_\rho^{-1} = \frac{1}{1-\rho^2}\begin{smallbmatrix}1 & -\rho \\ -\rho & 1 \end{smallbmatrix}$ and under $\mathbb{P}_{1|\sigma}$, $$X_i\,|Y_{\sigma_i}\iid\mathcal{N}(\rho Y_{\sigma_i}, (1-\rho^2)\id_d) \text{ for } i \in [n] \quad \text{and} \quad  X_i \perp \!\!\! \perp Y_j \text{ for } j\neq \sigma_i \ .$$
\subsection{Subexponential random variable}
First, we recall the definition of subexponential random variables  (see, for example, \cite[Proposition 2.9]{wainwright2019high}), which we will need for the square of subgaussian.
\begin{definition}\label{def: subexp}
	$X$ is $\mathrm{subexponential}(\nu,b)$  with nonnegative  $\nu,b$ if the following holds:
	$$\expe \exp({\delta(X-\expe X)} )\leq \exp\left(\frac{\nu^2\delta^2}{2}\right) \text{ for all } \abs{\delta}< \frac{1}{b} \ .$$
	Moreover, it has the tail bound for $t>0$,
	$$\mathbb{P}(X - \expe X \geq t)\leq 
	\exp\left(-\frac{t}{2}\min\left\{\frac{t}{\nu^2},\frac{1}{b} \right\}\right)  \ .$$
\end{definition}

Now, we present two inequalities for the concentration of square of Gaussian random variable are presented. Lemma~\ref{laurent-massart} is a tighter version of subexponential inequality for the sum of independent Gaussian squares whereas Lemma~\ref{gaussian-chaos} is a special case of Hanson-Wright inequality.
\begin{lemma}[Laurent-Massart {\cite[Lemma~1]{laurent-massart}} ]\label{laurent-massart}
	Let $X_1,\ldots, X_d \iid\mathcal{N}(0,1)$ and $\alpha=(\alpha_1,\ldots, \alpha_d) \in \mathbb{R}^{d}$ with nonnegative entries. For $t>0$, 
	
	$$\mathbb{P}\left( \sum_{i=1}^d \alpha^{\vphantom{2}}_i(X_i^2-1) \leq - t\right) \leq \exp\left(-\frac{t^2}{4\norm{\alpha}_2^2}\right)$$
\end{lemma}

\begin{lemma}[Gaussian Chaos]\label{gaussian-chaos}
	$X \sim \mathcal{N}(0,\id_d)$, $A=O+D$ symmetric matrix where $D = \text{diag}(\alpha)$ with $\alpha_i=A_{ii}$ for $i\in[d]$ and $O$ is the off-diagonal part of $A$  with $\lambda=(\lambda_1,\ldots,\lambda_d)$ eigenvalues of $O$. For $t>0$,
	$$\mathbb{P} \left(X^\mathsf{T}AX - \expe X^\mathsf{T}AX \geq t \right) \leq 2\exp\left(-\frac{t}{16}\min\left\{\frac{t}{2\norm{\alpha}_2^2}, \, \frac{1}{ \norm{\alpha}_\infty^{\vphantom{2}}} ,\, \frac{t}{2\norm{\lambda}_2^2}, \, \frac{1}{ \norm{\lambda}_\infty^{\vphantom{2}} }\right\} \right)$$
\end{lemma}
\begin{proof}
	First, we separate it into diagonal and off-diagonal parts as
	\begin{align*}
		X^\mathsf{T}AX - \expe X^\mathsf{T}AX &= X^\mathsf{T}OX - \expe X^\mathsf{T}OX + X^\mathsf{T}DX - \expe X^\mathsf{T}DX\\
		& = \sum_{i\neq j} O_{ij}(X_iX_j - \expe X_i X_j) +\sum_{i=1}^d \alpha^{\vphantom{2}}_i(X_i^2 -\expe X_i^2) \\ 
		&= \sum_{i\neq j} O_{ij}X_iX_j +\sum_{i=1}^d \alpha^{\vphantom{2}}_i(X_i^2 -1)
	\end{align*}
	Since $\{W+Z \geq t\} \subseteq \{W\geq \gamma t\}\cup\{Z\geq (1-\gamma)t\} $ for any $W,Z$ and $\gamma \in [0,1]$, by union bound with $\gamma=\half$
	$$\mathbb{P} \left(X^\mathsf{T}AX - \expe X^\mathsf{T}AX \geq t \right)  \leq \, \mathbb{P} \bigg(\sum_{i\neq j} O_{ij}X_iX_j  \geq \frac{t}{2} \bigg)  + \mathbb{P} \bigg(\sum_{i} \alpha^{\vphantom{2}}_i(X_i^2 -1) \geq \frac{t}{2} \bigg)  \ .$$
	
	By a simple calculation, one can show that $X_1^2$ is $\mathrm{subexponential}(2,4)$ which implies that $\sum \alpha^{\vphantom{2}}_i(X_i^2 -1)$ is $\mathrm{subexponential}\left(2\norm{\alpha}_2 ,4\norm{\alpha}_\infty \right)$. Therefore, by Definition~\ref{def: subexp}, the first summand can be bounded as
	\begin{equation}\label{eq: gaussian-chaos-first}
		\mathbb{P} \left(\sum_{i=1}^d \alpha_i(X_i^2 -1) \geq \frac{t}{2} \right) \leq \exp\left(-\frac{t}{16}\min\left\{\frac{t}{2\norm{\alpha}_2^2}, \, \frac{1}{ \norm{\alpha}}_\infty \right\}\right) \ .
	\end{equation} 
	
	For the other part, since $O$ is symmetric, it's diagonalizable and we can apply eigendecomposition as $O=B^\mathsf{T}\Lambda B$ where $\lambda_i$s are the eigenvalues of $O$ forming the matrix $\Lambda = \mathrm{diag}(\lambda_1,\ldots,\lambda_d)$ and $B$ is the orthonormal matrix having the columns as the corresponding eigenvectors. Hence, we can rewrite the off-diagonal random part as
	
	$$X^\mathsf{T}OX = \sum_{i=1}^d \lambda_i \Big(\sum_{j=1}^d B_{ij}X_j\Big)^2 \stackrel{\mathrm{d} }{=} \sum_{i=1}^d \lambda^{\vphantom{2}}_iX_i^2 = \sum_{i=1}^d \lambda^{\vphantom{2}}_i (X_i^2-1)$$
	where we used the facts $ \sum_{j}B_{ij}X_j\iid \mathcal{N}(0,1)$ for $i\in [n]$ and $\sum_{i} \lambda_i = \mathrm{tr}(\Lambda) = \mathrm{tr}(\Lambda B B^{-1}) = \mathrm{tr}(B^{-1}\Lambda B ) = \mathrm{tr}(O)=0$ (As a side note, this implies that a symmetric matrix with zero diagonal cannot be positive definite.). Hence, again 
	$X^\mathsf{T}OX = \sum_{i=1}^d \lambda^{\vphantom{2}}_i (X_i^2-1)$  is $\mathrm{subexponential}\left(2\norm{\lambda}_2 ,4\norm{\lambda}_\infty \right)$,  and we have 
	
	\begin{equation}\label{eq: gaussian-chaos-second}
		\mathbb{P} \left(X^\mathsf{T}OX  \geq \frac{t}{2} \right) \leq \exp\left(-\frac{t}{16}\min\left\{\frac{t}{2\norm{\lambda}_2^2}, \, \frac{1}{ \norm{\lambda}_\infty^{\vphantom{2}} }\right\}\right)  \ .
	\end{equation}
	By summing \eqref{eq: gaussian-chaos-first} with \eqref{eq: gaussian-chaos-second}  and using  $e^{-a} + e^{-b} \leq 2e^{-\min\{a,b\}}$, we get the desired result. 
\end{proof}

For self-containment of the paper, the moment generating functions of two functions of Gaussian random vectors are given below.
\begin{lemma}\label{lemma: gaussian-mgf}
	For $X \sim \mathcal{N}(0,\id_d)$, $b\in \mathbb{R}^d$, $\id_d-R$ positive definite and symmetric,   $$\expe\,{\exp \left( \half X^\mathsf{T}R X+X^\mathsf{T} b \right) }=  \exp\left(\half{b^\mathsf{T}(\id_d-R)^{-1}b}\right) \mathrm{det}(\id_d-R)^{-1/2} \ .$$
\end{lemma}
\begin{proof}
	Since $\id_d-R$ is positive definite and symmetric, by Cholesky decomposition, there exists a matrix $U$ such that $\id_d-R = U^\mathsf{T}U$. 
	Now, let us complete the square for the following expression:
	\begin{align}
		\half x^\mathsf{T}R x+x^\mathsf{T} b - \half x^\mathsf{T}x &= -\half \left( x^\mathsf{T}(\id_d-R)x  - 2x^\mathsf{T} b \right) \nonumber \\ 
		& = -\half \left( \norm{Ux}^2  - 2(Ux)^\mathsf{T} U^{-\mathsf{T}} b \right) \nonumber \\ 
		& = -\half \left( \norm{Ux-U^{-\mathsf{T}} b}^2 - \norm{U^{-\mathsf{T}} b}^2 \right)  \label{eq: mgf-gauss-middle}
	\end{align}
	By probability density function of $X$ and \eqref{eq: mgf-gauss-middle},
	\begin{align*}
		\expe\,{\exp \left(\half X^\mathsf{T} RX+X^\mathsf{T} b \right) } & =  \int_{-\infty}^\infty\ldots \int_{-\infty}^\infty (2 \pi)^{-d/2} \exp\left( \half x^\mathsf{T}R x+x^\mathsf{T} b - \half x^\mathsf{T}x \right) \dif x\\
		&  = \int_{-\infty}^\infty\ldots \int_{-\infty}^\infty (2 \pi)^{-d/2}  \exp \left( -\frac{   \norm{Ux-U^{-\mathsf{T}} b}^2}{2} + \frac{ \norm{U^{-\mathsf{T}} b}^2}{2}  \right)\dif x \\
		& = \exp \left(  \half{b^\mathsf{T}\left(U^\mathsf{T}U\right)^{-1}b } \right)  \int_{-\infty}^\infty \ldots \int_{-\infty}^\infty (2 \pi)^{-d/2} \\
		& \qquad \qquad  \cdot  \exp \left( -\half \left(x-\left(U^\mathsf{T}U\right)^{-1}b \right)^\mathsf{T} \left(\left(U^\mathsf{T}U\right)^{-1}\right)^{-1} \left(x-\left(U^\mathsf{T}U\right)^{-1}b \right) \right)\dif x \\
		&=\exp \left(  \half{b^\mathsf{T}\left(\id_d-R\right)^{-1}b } \right) \mathrm{det}\left((\id_d-R)^{-1}\right)^{1/2}
	\end{align*}
	where $\id_d-R \succ 0$ is needed for the integral to converge.
\end{proof}

\begin{lemma}\label{lemma: gaussian-mgf-2}
	For $X,Y \iid \mathcal{N}(0,\id_d)$ and $1+2a > |b|$,   
	$$\,\expe \exp\left(-a\norm{X}^2 -a\norm{Y}^2 +bX^\mathsf{T} Y\right) =  \left((1+2a)^2 -b^2\right)^{-d/2} \ .$$
\end{lemma}
\begin{proof}
	We have 
	\begin{flalign*}
\expe_{X,Y \, \iid \, \mathcal{N}(0,\id_d)}\exp\left(-a\norm{X}^2 -a\norm{Y}^2 +bX^\mathsf{T} Y\right)& \stackrel{\text{(i)}}{=} \mathbb{E}_Y \left[\exp\left(-a\norm{Y}^2 \right)\cdot\expe_X\Big[\exp\left(-a\norm{X}^2 +bX^\mathsf{T} Y\right) \,\middle| \,Y\,\Big] \, \right]\\
	&\stackrel{\text{(ii)}}{=}(1+2a)^{-d/2}\cdot \expe_Y \bigg[\exp\left(\norm{Y}^2\left(-a + \frac{b^2}{2(1+2a)}\right)\right)\bigg]\\
	&\stackrel{\text{(iii)}}{=}(1+2a)^{-d/2}\left(1-2\left(-a + \frac{b^2}{2(1+2a)}\right)\right)^{-d/2}\\
	&=\left((1+2a)^2 -b^2\right)^{-d/2} 
\end{flalign*}
where (i)  is by the law of total expectation, (ii) and (iii) are by Lemma~\ref{lemma: gaussian-mgf} with $1+2a>0$ is needed for the first integral and $|1+2a|>|b|$ is needed for the second integral to converge.
\end{proof}

\subsection{Some results about permutations}\label{app: perms}

In this section, we state some technical details about permutations which have a similar flavor as those in \cite{jiaming2021testinggraphs}. For a permutation $\sigma$, let $\mathcal{C}_\sigma$ be the \textit{set of its cycles} and for any $ C \in \mathcal{C}_\sigma$ let $\mathcal{O}_C$ be the \textit{orbit }of cycle $C $ of length $|C| \in [n]$. Moreover, if $N_k^\sigma$ is the number of $k$-cycles at $\sigma$, then $N_k^\sigma \in \left[ \floor{\frac{n}{k}}\right] \cup \{0\}$ and $\sum_{k \in [n]} k\cdot N_k^\sigma  = n$ should hold.

For a fixed $\sigma$ and $i\in[n]$, denote $Z_i$ as the product of likelihood ratios of pairs $(X_i,Y_i)$ and $(X_i,Y_{\sigma_i})$, i.e.,
$$Z_{i} =\frac{\mathcal{N}_{\rho}\left(Y_{\sigma_i} \middle|  X_i\right)}{\mathcal{N}\left(Y_{\sigma_i} \right)} \cdot \frac{\mathcal{N}_{\rho}\left(Y_{i \vphantom{\sigma_i}} \middle|  X_i\right)}{\mathcal{N}\left(Y_{i \vphantom{\sigma_i}}  \right)} $$ and let $Z_C$ be the product of the $Z_i$ in cycle $C\in \mathcal{C}_\sigma$ which turns out to be 
$$Z_C =  \prod_{i \in \mathcal{O}_{C}} \, Z_i= \frac{1}{(1-\rho^2)^{d|C|}} \exp\left(-\frac{1}{1-\rho^2}\sum_{i \in \mathcal{O}_C}\left( \rho^2 \norm{Y_i}^2  + \rho^2 \norm{X_i}^2 -\rho X_i^\mathsf{T}(Y_i+Y_{\sigma_i}) \right)\right) \ .$$
Observe that for independent random variables $X_1, \ldots, X_n,Y_1, \ldots , Y_n$  and different cycles $C_1,C_2 \in \mathcal{C}_\sigma $ of a given permutation $\sigma\in S_n$, since the cycles are disjoint, we observe that $Z_{C_1}$ and $Z_{C_2}$ do not possess a common random variable. Hence,  $Z_{C_1} \perp \!\!\! \perp Z_{C_2}$.

The expectation of $Z_C$ under $H_0$ is computed as following. For a different version of the proof, follow the steps of { \cite[Proposition 1]{jiaming2021testinggraphs}}.
\begin{lemma}\label{lemma: cycle-exp}
	$\mathbb{E}_{0}\, Z_C =\left(1-\rho^{2 |C|}\right)^{-d}$.
\end{lemma}
\begin{proof}
	To simplify the expression, we first condition on one of the databases, without loss of generality condition on $Y^n$. By applying the independence of $X_i$s under $H_0$ and the moment generating functions in Lemma~\ref{lemma: gaussian-mgf}, we get
	\begin{flalign*}
		\expe_0 \left[ Z_C \, |Y^n \right] = & \, \frac{1}{(1-\rho^2)^{d|C|}} \exp\left(-\frac{\rho^2}{1-\rho^2}\sum_{i \in \mathcal{O}_C} \norm{Y_i}^2 \right)  \\
		& \qquad \qquad \cdot \prod_{i \in \mathcal{O}_C}\expe_0 \left[ \exp\left( -\frac{1}{1-\rho^2}\left(  \rho^2 \norm{X_i}^2 -\rho X_i^\mathsf{T}(Y_i+Y_{\sigma_i}) \right)\right)\middle|\, Y^n\right] \\
		=& \,\frac{1}{(1-\rho^2)^{d|C|}} \exp\left(-\frac{\rho^2}{1-\rho^2}\sum_{i \in \mathcal{O}_C} \norm{Y_i}^2 \right)  \\
		& \qquad \qquad \cdot \left(1 + 2\frac{\rho^2}{1-\rho^2}\right)^{-d|C|/2}\exp \left(\frac{\frac{\rho^2}{(1-\rho^2)^2}}{2\left(1 + 2\frac{\rho^2}{1-\rho^2}\right)}\sum_{i \in \mathcal{O}_C}\norm{Y_i + Y_{\sigma_i}}^2\right)  \\
		= &\, (1-\rho^4)^{-\frac{d|C|}{2}}\exp\left(\sum_{i \in \mathcal{O}_C}-\frac{\rho^4}{1-\rho^4} \norm{Y_i}^2 + \frac{\rho^2}{1-\rho^4}\,Y_i^\mathsf{T}Y_{\sigma_i}\right) \ .
	\end{flalign*}
	Therefore, by the law of total expectation,
	\begin{flalign}
		(1-\rho^4)^{\frac{d|C|}{2}}\, \expe_0\, Z_C  & = (1-\rho^4)^{\frac{d|C|}{2}}\, \expe_0\,\expe_0\left[ Z_C \, | Y^n\right] \nonumber\\ &= \expe_{Y_i \, \iid\, \mathcal{N}(0,\id_d)} \exp\left(\frac{\rho^2}{1-\rho^4} \sum_{i \in \mathcal{O}_C} -\rho^2\norm{Y_i}^2 + Y_i^\mathsf{T}Y_{\sigma_i}\right) \nonumber\\
		&= \left(\expe_{Y_i\, \iid\, \mathcal{N}(0,1)} \exp\left(\frac{\rho^2}{1-\rho^4} \sum_{i \in \mathcal{O}_C} -\rho^2 Y_i^2 + Y_iY_{\sigma_i}\right) \right)^d \nonumber\\
		& =\left(\expe_{Y\sim\, \mathcal{N}(0,\id_{|C|})}\exp \left(\half Y^\mathsf{T}RY\right)\right)^d = \left(\det(\id_{|C|}-R)\right)^{-\frac{d}{2}}\label{eq: cycle-circulant}
	\end{flalign}
	where $R$ is a circulant matrix  with entries
	$$R = \frac{\rho^2}{1-\rho^4}\begin{bmatrix}
		-2\rho^2 & 1 & 0 & \cdots & 1 \\
		1 & -2\rho^2 & 1 & \cdots & 1 \\
		0 & 1 & -2\rho^2 &\cdots & 0\\
		0 & 0 & 1  & \cdots & 0\\
		\vdots & \vdots&\vdots & \ddots& \vdots \\
		0 & 0 & 0 & \cdots & 0 \\
		1  & 0 & 0 & \cdots  & -2\rho^2
	\end{bmatrix}_{|C|\times |C|}$$
which we represent as $R = \frac{\rho^2}{1-\rho^4}(	-2\rho^2, 1 ,0 , \ldots, 0 ,1)$ and the last step comes from Lemma~\ref{lemma: gaussian-mgf}.

	Since $\id_{|C|}-R$ is circulant $\frac{1}{1-\rho^4}(	1+\rho^4, -\rho^2 ,0 , \ldots, 0 ,-\rho^2)$, its determinant can be calculated by using the Fourier coefficients as  ($i^2=-1$) (see \cite[Theorem 1]{circulant})
	\begin{flalign*}
		\det(\id_{|C|}-R) = &\prod_{k=0}^{|C|-1} \frac{1+\rho^4}{1-\rho^4} -\frac{\rho^2}{1-\rho^4}e^{\frac{k}{|C|}2\pi i} - \frac{\rho^2}{1-\rho^4}e^{\frac{k(|C|-1)}{|C|}2\pi i}\\
		=& \,(1-\rho^4)^{-|C|} \prod_{k=0}^{|C|-1} 1+\rho^4-\rho^2\left(e^{\frac{k}{|C|}2\pi i} +e^{-\frac{k}{|C|}2\pi i} \right)\\
		=& \,(1-\rho^4)^{-|C|}  \prod_{k=0}^{|C|-1} \left(1-\rho^2 e^{\frac{k}{|C|}2\pi i}\right)\left(1-\rho^2 e^{-\frac{k}{|C|}2\pi i}\right) = (1-\rho^4)^{-|C|}(1-\rho^{2|C|} )^2 
	\end{flalign*}
where we used that  one can write $1-z^n = \prod_{k=0}^{n-1}(1-z\cdot z_k^{-1})$ with $z_k =e^{\frac{k}{n} 2 \pi i} $ for $k=0,\ldots,n-1$ are $n$ complex roots of  $1=z^n$.
	By putting this into \eqref{eq: cycle-circulant}, we get the desired result.
\end{proof}

\subsubsection{Derangements}
As a side note, let us state two technical combinatoric facts. See, for example, \cite{concretemath} for more details.

First, from Stirling's approximation, we have the following upper bound for binomial coefficient: 
$$\binom{n}{k} \leq \left(\frac{e n}{k}\right)^k \ .$$

 Secondly, let $!n$ be the number of derangements  over $[n]$, i.e., permutations without a fixed point which is defined as $!n = n! \sum_{k=0}^n \frac{(-1)^k}{k!}$. Then, we have  $\frac{!n}{n!} = \sum_{k=0}^n \frac{(-1)^k}{k!}\leq 1$ for $n\geq1$. For $\sigma \sim \mathrm{Unif}(S_n)$, the probability of its number of fixed points being equal to $k\in[n-1]\cup \{0\}$ is $$\mathbb{P}\left[ N_1^\sigma = k\right] = \binom{n}{k}\frac{!(n-k) }{n!}\leq \binom{n}{k} \frac{(n-k)!}{n!} = \frac{1}{k!}$$ which is a decreasing function in $k$. 

\subsection{Second-order moment method \texorpdfstring{\boldmath{$\mathbb{E}_0 \,{L^2}$}}{E0L2}}\label{app: unconditional-converse}
\begin{proof}[Proof of Lemma~\ref{lemma: converse-unconditional}]
	For any fixed $\sigma \in S_n$, we have
	\begin{flalign}
		\mathbb{E}_{0} \, \prod_{i=1}^n \, Z_i
		&\stackrel{\hphantom{\text{(iii)}}}{=}  \mathbb{E}_{0}\, \prod_{C \in \mathcal{C}_\sigma} \, \prod_{i \in \mathcal{O}_C} \, Z_i \nonumber\\
		&\stackrel{\hphantom{\text{(iii)}}}{=} \mathbb{E}_{0} \prod_{C \in \mathcal{C}_\sigma}  Z_C \nonumber\\
		&\stackrel{\hphantom{\,}\text{(i)\hphantom{\,}}}{=} \prod_{C \in \mathcal{C}_\sigma} \mathbb{E}_{0} \, Z_C \nonumber\\
		&\stackrel{\hphantom{\,}\text{(ii)\hphantom{ }}}{=}   \prod_{k=1}^n \,\prod_{\substack{C \in \mathcal{C}_\sigma :\\ k\text{-cycle}}} \,\frac{1}{\left(1-\rho^{2 k}\right)^d}\nonumber \\
		&  \stackrel{\hphantom{\text{(iii)}}}{=}  \prod_{k=1}^n \frac{1}{\left(1-\rho^{2 k}\right)^{dN_k^\sigma}}\nonumber\\
		&\stackrel{\text{(iii)}}{\leq}   \prod_{k=1}^n \frac{1}{\left(1-\rho^{2 }\right)^{dkN_k^\sigma}} \nonumber\\
		& \stackrel{\hphantom{\text{(iii)}}}{=} \frac{1}{\left(1-\rho^{2 }\right)^{dn}} \label{eq: cycle-exp-result}
	\end{flalign}
	where (i) is due to independence $Z_{C_1} \perp \!\!\! \perp Z_{C_2}$ of different cycles $C_1,C_2 \in \mathcal{C}_\sigma $, (ii) is by Lemma~\ref{lemma: cycle-exp} and (iii) is by the chain of inequalities $(1-x)^k \leq 1-x \leq 1-x^k$ for $k\geq 1$ and $x \in [0,1]$.
	 Applying  Ingster-Suslina method Lemma \ref{lemma: ingster-suslina}, we get a reformulation of the second-order moment of the likelihood ratio under $H_0$ as
	\begin{align}\label{eq: second-order-middle}
		\expe_0\, L^2 &= \expe_{\sigma, \tilde{\sigma} \, \iid \, \pi} \, G(\sigma,\tilde{\sigma})  = \expe_{\tilde{\sigma}\sim \pi}\, \expe_{\sigma \sim \pi}\left[\, G(\sigma,\tilde{\sigma})\,|\, \tilde{\sigma}\right] 
\end{align}
with $$G(\sigma,\tilde{\sigma})=  \mathbb{E}_0 \left[ \, \frac{\mathbb{P}_{1|\sigma}(X^n,Y^n)\cdot\mathbb{P}_{1|\tilde{\sigma}}(X^{n},Y^{n})}{\mathbb{P}_0^2(X^n,Y^n)} \,\right]\ .$$
	Since any permutation is bijective and the prior distribution is set as $\pi = \mathrm{Unif}(S_n)$, $$\expe_{\sigma} [\, G(\sigma,\tilde{\sigma}) \,| \, \tilde{\sigma}\,]=\expe_{\sigma} [\, G(\tilde{\sigma}^{-1}\circ\sigma,\mathrm{id}) \,| \, \tilde{\sigma}\,]=\expe_{\sigma} [\, G(\sigma,\mathrm{id}) \,| \, \tilde{\sigma}\,]=\expe_{\sigma} [\, G(\sigma,\mathrm{id})]$$ holds for any fixed $\tilde{\sigma} \in S_n$. Hence, WLOG, we set $\tilde{\sigma}=\mathrm{id}$. By plugging  this observation into \eqref{eq: second-order-middle}, we have
\begin{align*}
		\expe_0\, L^2& =  \expe_{\sigma \sim \pi}\, \mathbb{E}_0\left[ \,\frac{\mathbb{P}_{1|\sigma}(X^{n},Y^{n})\cdot\mathbb{P}_{1|\text{id}}(X^{n},Y^{n})}{\mathbb{P}_0^2(X^{n},Y^{n})} \, \right] \\
		&=\expe_{\sigma \sim \pi}\,  \mathbb{E}_{0}\left[\, \mathlarger{\prod_{i=1}^n} \,
		\frac{\mathcal{N}_{\rho}\left(X_i, Y_{\sigma_i}\right)\cdot\mathcal{N}_{\rho}\left(X_i, Y_{i}\right)}{\mathcal{N}\left(X_i \right)^2\cdot\mathcal{N}\left(Y_i \right)^2}\,\right] \\
		& = \mathbb{E}_{\sigma \sim \pi}\,  \mathbb{E}_{0} \, \mathlarger{\prod_{i=1}^n} \, Z_i  \\
		&\stackrel{(*)}{\leq} \frac{1}{\left(1-\rho^{2 }\right)^{dn}}
\end{align*}
	where (*) follows \eqref{eq: cycle-exp-result}. So, $\expe_0\, L^2 \leq 1+o(1) \,\text{ if } \, dn\frac{\rho^2}{1-\rho^2} =o(1) $ which is implied by $\rho^2 =o\left( \frac{1}{dn}\right)$.
\end{proof}

\subsection{First-order moment of the truncated likelihood ratio \texorpdfstring{\boldmath{$\mathbb{E}_0 \,\tilde{L}$}}{E0L}}\label{app: truncated-first}
\begin{proof}[Proof of Lemma~\ref{lemma: truncated-first}]
	As mentioned in Section~\ref{sec: truncated}, it suffices to find conditions such that $\mathbb{P}_{1|\sigma}(\Gamma_\sigma)\geq 1-\mathcal{O}(1) $ for  any $\sigma \in S_n$. Fix $\sigma$ and upon applying union bounds, we get
	\begin{flalign*}
		1&-\mathbb{P}_{1|\sigma}(\Gamma_\sigma) = \mathbb{P}_{1|\sigma}(\,\Gamma_\sigma^C\,)\leq \underbrace{ \sum_{k=k^*}^{N_1^\sigma} \, \sum_{\substack{T \subseteq F_\sigma \\ |T|=k}} 2 \mathbb{P}_{1|\sigma}\bigg(\,\sum_{i \in T} \norm{X_i}^2 \leq w_k\bigg)}_{\triangleq\, \mathrm{(\mathcal{E}_1)}} +\,  \underbrace{ \sum_{k=k^*}^{N_1^\sigma} \, \sum_{\substack{T \subseteq F_\sigma \\ |T|=k}}\mathbb{P}_{1|\sigma}\bigg( \sign{\rho}\sum_{i \in T} X_i^\mathsf{T}Y_{\sigma_i}^{\vphantom{\mathsf{T}}} \geq v_k\bigg)}_{\triangleq\, \mathrm{(\mathcal{E}_2)} }  \ .
	\end{flalign*}
	
	\paragraph{Bounding  $\pmb{(\mathcal{E}_1)}$:} Let
	$$w_k = dk -2\sqrt{d}k\, r_k , \quad  \psi_1 = \min_{k=k^*,\ldots, N_1^\sigma}\, - \ln \frac{e n}{k} +r_k^2 $$ and $Z\sim\chi^2_{dk}$. Then,  for $0< w_k < dk$, we have 
	\begin{align}
		(\mathcal{E}_1)&=\sum_{k=k^*}^{N_1^\sigma} \sum_{\substack{T \subseteq F_\sigma \\ |T|=k}}2\mathbb{P}\left(\, Z \leq w_k\right) = 2 \sum_{k=k^*}^{N_1^\sigma}  \binom{N_1^\sigma}{k} \mathbb{P}\left(Z - dk \leq -(dk -w_k)\right)   \nonumber \\
		&\stackrel{\text{(i)}}{\leq}  2\sum_{k=k^*}^{N_1^\sigma} \binom{N_1^\sigma}{k} \exp\left(-\frac{(dk-w_k)^2}{4dk}\right) \nonumber\\
		&\stackrel{\text{(ii)}}{\leq} 2\sum_{k=k^*}^{N_1^\sigma} \exp\left(k\ln \frac{e n}{k}-\frac{(dk-w_k)^2}{4dk}\right)\nonumber \\
		& = 2\sum_{k=k^*}^{N_1^\sigma} \exp\left(-k \left(-\ln \frac{e n}{k} +r_k^2\right)\right) \nonumber \\
		& \leq  2\sum_{k=k^*}^{N_1^\sigma}  e^{-k\psi_1} \nonumber\\
		&\leq 2e^{-k^*\psi_1}\sum_{k=0}^{\infty} e^{-k\psi_1}\nonumber\\
		& \stackrel{\text{(iii)}}{\leq}  \frac{2e^{-k^*\psi_1}}{1-e^{-\psi_1}}\label{eq: first-moment-1}
	\end{align}
	where  (i) follows Laurent-Massart Lemma~\ref{laurent-massart}, (ii) is by the inequality $\binom{n}{k}\leq \left(\frac{en}{k}\right)^k$ from Stirling's approximation and due $N_1^\sigma \leq n$ for all $\sigma$; and in (iii), $\psi_1>0$ should hold for any $\sigma$. Note that the conditions $\psi_1 >0 $ and $w_k > 0$ are satisfied if  for all $k=k^*,\ldots, n$ $$\frac{\sqrt{d}}{2} > r_k >\sqrt{ \ln \frac{en}{k}} $$
	holds which requires $d > 4\ln \frac{en}{\;k^*}$ .
	
	\paragraph{Bounding  $\pmb{(\mathcal{E}_2)}$:} Let 
	$$v_k = |\rho| dk + 4|\rho| \sqrt{d}k \,s_k , \quad 
	\psi_2 = \min_{k=k^*,\ldots, N_1^\sigma} \frac{|\rho| \sqrt{d} \,s_k }{4}\min \left\{\frac{1}{|\rho|},\, \frac{2}{\sqrt{1-\rho^2}},\, \frac{s_k }{|\rho|\sqrt{d}},\, \frac{4|\rho| s_k }{(1-\rho^2)\sqrt{d}}\right\} - \ln \frac{e n}{k} $$
	with $s_k\geq 0$ and $Z_i \iid \mathcal{N}(0,\id_d)$ for $i \in [n]$ and independent from $Y^n$. Then, for $v_k > |\rho| kd$ and any fixed $T$ with $|T|=k$,
	\begin{align}
\mathbb{P}_{1|\sigma}\Bigg(\sign{\rho}\sum_{i \in T} X_i^\mathsf{T}Y_{\sigma_i}^{\vphantom{\mathsf{T}}} &  \geq v_k\Bigg) = \mathbb{P}\Bigg(\sum_{i \in T}|\rho| Y_{\sigma_i}^\mathsf{T}Y_{\sigma_i}^{\vphantom{\mathsf{T}}} +\sign{\rho}\sqrt{1-\rho^2}\, Y_{\sigma_i}^\mathsf{T} Z_i^{\vphantom{\mathsf{T}}}\geq v_k \Bigg) \nonumber\\
		& \stackrel{\text{(i)}}{=} \,\mathbb{P} \left(\xi^\mathsf{T}A\xi -|\rho| dk \geq v_k -|\rho| dk   \right)  \nonumber\\
		& \stackrel{\text{(ii)}}{\leq} \exp\left(-\frac{v_k-|\rho| dk}{16}\min \left\{\frac{1}{|\rho|},\, \frac{2}{\sqrt{1-\rho^2}},\, \frac{v_k-|\rho| dk}{2\rho^2dk},\, \frac{v_k-|\rho| dk}{(1-\rho^2)dk}\right\}\right)\label{eq: first-moment-2-mid}
	\end{align}
			where the variables 
			$$\xi =  \begin{bmatrix}
				Y_{\sigma_{i_1}}^\mathsf{T} & \cdots & Y_{\sigma_{i_k}}^\mathsf{T} & Z_{i_1}^\mathsf{T} & \cdots & Z_{i_k}^\mathsf{T}
			\end{bmatrix}^\mathsf{T} \quad \text{with } \left\{i_1,\ldots,i_k\right\}=T \ , $$
		$$ A = \frac{\sign{\rho}}{2} \begin{bmatrix}
			2\rho \id_{dk} & \sqrt{1-\rho^2}\, \id_{dk} \\
			\sqrt{1-\rho^2}\, \id_{dk}  & \vphantom{\frac{\sqrt{1-\rho^2}}{2}} \zeros_{dk}
		\end{bmatrix}$$ 
	are put in (i)  
	and Lemma~\ref{gaussian-chaos} is applied in (ii) with its notation for $A$ as 
			$$D = |\rho|\begin{bmatrix}  \id_{dk} & \zeros_{dk} \\ \zeros_{dk}  & \zeros_{dk} \end{bmatrix} \quad \text{with} \quad \norm{\alpha}_2^2 = \rho^2 dk, \quad \norm{\alpha}_{\infty} = |\rho| \quad \text{and}$$ 
			$$O = \frac{\sign{\rho} \sqrt{1-\rho^2}}{2}  \begin{bmatrix} \zeros_{dk}  & \id_{dk} \\
			\id_{dk}  & \zeros_{dk}
		\end{bmatrix}\quad \text{with} \quad \norm{\lambda}_2^2 = \frac{1-\rho^2}{4}2dk, \quad\norm{\lambda}_\infty = \frac{ \sqrt{1-\rho^2} }{2}$$ 
	where for $O$, we used the fact that any eigenvalue $\beta$ of a permutation matrix $P=\begin{smallbmatrix} \zeros  & \id \\
		\id  & \zeros
	\end{smallbmatrix}$  should abide $\beta^2=1$. 
	
		Therefore, for $v_k > |\rho| dk$ with $k=k^*,\ldots,n$, by using Equation~\eqref{eq: first-moment-2-mid}, we get
	\begin{align}
		(\mathcal{E}_2) &= \sum_{k=k^*}^{N_1^\sigma} \sum_{\substack{T \subseteq F_\sigma \\ |T|=k}}  \mathbb{P}_{1|\sigma}\left(\sign{\rho}\sum_{i \in T} X_i^\mathsf{T}Y_{\sigma_i}^{\vphantom{\mathsf{T}}} \geq v_k\right) 
		\nonumber \\
		& \leq 2\sum_{k=k^*}^{N_1^\sigma} \sum_{\substack{T \subseteq F_\sigma \\ |T|=k}}  \exp\left(-\frac{v_k-|\rho| dk}{16}\min \left\{\frac{1}{|\rho|},\, \frac{2}{\sqrt{1-\rho^2}},\, \frac{v_k-|\rho| dk}{2\rho^2dk},\, \frac{v_k-|\rho| dk}{(1-\rho^2)dk}\right\}\right) \nonumber\\
		&= 2\sum_{k=k^*}^{N_1^\sigma} \binom{N_1^\sigma }{k} \exp\left(-\frac{v_k-|\rho| dk}{16}\min \left\{\frac{1}{|\rho|},\, \frac{2}{\sqrt{1-\rho^2}},\, \frac{v_k-|\rho| dk}{2\rho^2dk},\, \frac{v_k-|\rho| dk}{(1-\rho^2)dk}\right\}\right) \nonumber\\
		& \stackrel{\text{(iii)}}{\leq} 2\sum_{k=k^*}^{N_1^\sigma} \exp\left(-k \left( -\ln \frac{e n}{k} + \frac{|\rho| \sqrt{d} \,s_k }{4}\min \left\{\frac{1}{|\rho|},\, \frac{2}{\sqrt{1-\rho^2}},\, \frac{s_k }{|\rho|\sqrt{d}},\, \frac{4|\rho| s_k }{(1-\rho^2)\sqrt{d}}\right\}\right)\right)\nonumber\\
		& \leq 2\sum_{k=k^*}^{N_1^\sigma} e^{-k\psi_2} \nonumber\\
		& \leq 2e^{-k^*\psi_2}\sum_{k=0}^{\infty} e^{-k\psi_2} \nonumber\\
		&\stackrel{\text{(iv)}}{\leq}  \frac{2e^{-k^*\psi_2}}{1-e^{-\psi_2}} \label{eq: first-moment-2}
	\end{align}
	where (iii) is by the inequality $\binom{n}{k}\leq \left(\frac{en}{k}\right)^k$ from Stirling's approximation and due $N_1^\sigma \leq n$ for all $\sigma$; and  in (iv), $\psi_2 >0$ should hold for any $\sigma$. Note that the condition $\psi_2 >0$  is satisfied for any $\sigma \in S_n$ if for all  $k=k^*,\ldots,n$ 
	\begin{align*} 
		&\frac{|\rho| \sqrt{d} \,s_k }{4}\min \left\{\frac{1}{|\rho|},\, \frac{2}{\sqrt{1-\rho^2}},\, \frac{s_k }{|\rho|\sqrt{d}},\, \frac{4|\rho| s_k }{(1-\rho^2)\sqrt{d}}\right\} >  \ln \frac{e n}{k}  \\
		\iff & s_k > \frac{2}{\sqrt{d}}\max\left\{2, \frac{\sqrt{1-\rho^2}}{|\rho|}\right\} \ln \frac{e n}{k} \quad \& \quad s_k > \max\left\{2, \frac{\sqrt{1-\rho^2}}{|\rho|} \right\} \sqrt{\ln \frac{e n}{k}}\\
		\iff & s_k > \sqrt{\ln \frac{en}{k}}\max\left\{2, \frac{\sqrt{1-\rho^2}}{|\rho|}\right\}  
	\end{align*}
	holds where in the last step, we applied the condition $d \geq 4 \ln \frac{en}{\;k^*}$ needed for bounding $(\mathcal{E}_1)$ in the previous part. 
	By summing \eqref{eq: first-moment-1} and \eqref{eq: first-moment-2}, we get
	$$1-\mathbb{P}_{1|\sigma}(\Gamma_\sigma)\leq (\mathcal{E}_1)+(\mathcal{E}_2) \leq\frac{2e^{-k^*\psi_1}}{1-e^{-\psi_1}} +  \frac{2e^{-k^*\psi_2}}{1-e^{-\psi_2}} \leq  \frac{4e^{-k^*\min\{\psi_1,\psi_2\}}}{1-e^{-\min\{\psi_1,\psi_2\}}} \ . $$ 
	Consequently, $\mathbb{E}_0 \,\tilde{L}\to 1$ if  $\min\{\psi_1,\psi_2\} >0$ and $k^*\to \infty$ hold.
\end{proof}

\subsection{Second-order moment of the truncated likelihood ratio  \texorpdfstring{\boldmath{$\mathbb{E}_0 \,\tilde{L}^2$}}{E0L2}}\label{app: truncated-second}
\begin{proof}[Proof of Lemma~\ref{lemma: truncated-second}]
	We will follow the similar steps of Lemma~\ref{lemma: converse-unconditional} by examining the behavior of $\tilde{L} ^2$'s mean under $H_0$ for small and large number of fixed points of a permutation separately. From Ingster-Suslina method (Lemma~\ref{lemma: ingster-suslina}), we have
	\begin{equation} \label{eq: tuncated-second-beginning}
		\expe_0 \,\tilde{L} ^2  =\,\expe_{\sigma,\tilde{\sigma}\,\iid \, \pi}\, \expe_0 \,L_\sigma L_{\tilde{\sigma}} \,\mathbbm{1}_{\Gamma_\sigma \cap \Gamma_{\tilde{\sigma}}} 
	\end{equation}
	Now, we observe that 
\begin{flalign*}
	\Gamma_\sigma \cap \Gamma_{\tilde{\sigma}} & = \bigcap_{k=k^*}^{\min(N_1^\sigma, N_1^{\tilde{\sigma}})}\,\bigcap_{\substack{T \subseteq F_\sigma\cap F_{\tilde{\sigma}} \\ |T|=k}} \bigg\{\,\sum_{i\in T} X_i^\mathsf{T} X^{\vphantom{\mathsf{T}}}_i > w_k ,\,\sum_{i\in T} Y_i^\mathsf{T} Y^{\vphantom{\mathsf{T}}}_i > w_k , \,\sign{\rho}\sum_{i\in T} X_i^\mathsf{T}Y^{\vphantom{\mathsf{T}}}_{\sigma_i} < v_k\bigg\} \\
	& \subseteq \;\bigg\{\sum_{i\in F_\sigma} X_i^\mathsf{T} X^{\vphantom{\mathsf{T}}}_i > w_{N_1^\sigma} ,\, \sum_{i\in F_\sigma} Y_i^\mathsf{T} Y^{\vphantom{\mathsf{T}}}_i > w_{N_1^\sigma}  ,\, \sign{\rho}\sum_{i\in F_\sigma} X_i^\mathsf{T}Y^{\vphantom{\mathsf{T}}}_{\sigma_i} < v_{N_1^\sigma}\bigg\}\\
	& \triangleq  \;  \mathcal{E}_{F_\sigma} \ .
\end{flalign*}
By continuing from~\eqref{eq: tuncated-second-beginning} and choosing the prior distribution as $\pi = \mathrm{Unif}(S_n)$
\begin{equation}\label{eq: converse-truncated-second-likelihood}
\expe_0 \,\tilde{L} ^2	\leq \, \expe_{\sigma,\tilde{\sigma}\, \iid \,\mathrm{Unif}(S_n)}\, \expe_0 \,L_\sigma L_{\tilde{\sigma}} \mathbbm{1}_{\mathcal{E}_{F_\sigma}} = \, \expe_{\sigma \sim\mathrm{Unif}(S_n)} \,  \expe_0 \,L_\sigma L_{\mathrm{id}} \mathbbm{1}_{\mathcal{E}_{F_\sigma}}
\end{equation}
where the last step is due to the uniformity and symmetric behavior of $L_{\tilde{\sigma}}$ for all $\tilde{\sigma}\in S_n$ (See Equation~\eqref{eq: second-order-middle} and its following for a similar argument.). 
	
	Recall that for a fixed $\sigma$ and $i\in[n]$, we have defined
	$$Z_{i} =\frac{\mathcal{N}_{\rho}\left(Y_{\sigma_i} \middle|  X_i\right)}{\mathcal{N}\left(Y_{\sigma_i} \right)} \cdot \frac{\mathcal{N}_{\rho}\left(Y_{i \vphantom{\sigma_i}} \middle|  X_i\right)}{\mathcal{N}\left(Y_{i \vphantom{\sigma_i}}  \right)} $$ and  $Z_C$ is the product of $Z_i$s in cycle $C\in \mathcal{C}_\sigma$. For any $\sigma$, we can convert the second expectation of~\eqref{eq: converse-truncated-second-likelihood} as
\begin{flalign}
		\expe_0 \,L_\sigma L_{\mathrm{id}} \mathbbm{1}_{\mathcal{E}_{F_\sigma}} & = \expe_0 \prod_{C \in \mathcal{C}_\sigma} Z_C \, \mathbbm{1}_{\mathcal{E}_{F_\sigma}} =   \expe_0\prod_{\substack{C \in \mathcal{C}_\sigma \\ |C|>1} } Z_C  \cdot\prod_{i \in F_\sigma}Z_i \mathbbm{1}_{\mathcal{E}_{F_\sigma}} =  \prod_{\substack{C \in \mathcal{C}_\sigma \\ |C|>1} } \expe_0 \,Z_C  \cdot \expe_0 \prod_{i \in F_\sigma}Z_i\mathbbm{1}_{\mathcal{E}_{F_\sigma}} \label{eq: second-truncated-mid}
\end{flalign}
	where the fact $Z_{C_1}\perp \!\!\! \perp Z_{C_2}$ for different cycles $C_1 \neq C_2$ and the fact that $\mathcal{E}_{F_\sigma}$  is an event dependent on $Z_C$ for $|C|=1$ are used.  As covered in Appendix~\ref{app: perms}, let $N_k^\sigma$ is the number of $k$-cycles of $\sigma$ and we have that $ \sum_{k=2}^n k\cdot N_k^\sigma = n-N_1^\sigma $. 
	
	By applying Lemma \ref{lemma: cycle-exp}, the first product of~\eqref{eq: second-truncated-mid} is upper bounded by
	\begin{flalign}
		\prod_{\substack{C \in \mathcal{C}_\sigma \\ |C|>1} } \expe_0 \,Z_C = \prod_{k=2}^n \left(1-\rho^{2k}\right)^{-dN_k^\sigma} \leq \prod_{k=2}^n \left(1-\rho^{4}\right)^{-d\frac{kN_k^\sigma}{2}} = \left(1-\rho^{4}\right)^{-d\frac{n-N_1^\sigma}{2}}\label{eq: converse-truncated-second-product1}
	\end{flalign}
where the inequality comes from $\left(1-\rho^4\right)^{k/2} \leq 1-\rho^4 \leq 1-\rho^{4k/2}$ since $\rho^4 \in [0,1]$ and $k/2 \geq 1$.

	Moreover, the second product of~\eqref{eq: second-truncated-mid} can be written as 
	\begin{flalign}
		\expe_0 \prod_{i \in F_\sigma}  Z_i \, \mathbbm{1}_{\mathcal{E}_{F_\sigma}}  &=\expe_{\substack{X_i,Y_j \,\iid \,\mathcal{N}(0,\id_d) \\ \forall \, i,j \in F_\sigma}} \left[\,\prod_{i \in F_\sigma} \left(\frac{\mathcal{N}_\rho(Y_i \,| X_i)}{\mathcal{N}(Y_i)}\right)^2\mathbbm{1}_{\mathcal{E}_{F_\sigma}}\right] \nonumber\\
		& = (1-\rho^2)^{-dN_1^\sigma}\,\expe_{\substack{X_i,Y_j \,\iid \,\mathcal{N}(0,\id_d) \\ \forall\, i,j \in F_\sigma}}  \left[ \, \prod_{i \in F_\sigma} \exp\left( -\frac{\norm{Y_i -\rho X_i}^2}{1-\rho^2} + \norm{Y_i}^2 \right) \mathbbm{1}_{\mathcal{E}_{F_\sigma}} \right] \nonumber \\
		& = (1-\rho^2)^{-dN_1^\sigma}\, \expe_{X,Y \, \iid \, \mathcal{N}(0,\id_{dN_1^\sigma})} \bigg[ \exp\left(-\frac{\rho^2}{1-\rho^2}\Big(\norm{X}^2 + \norm{Y}^2 -\frac{2}{\rho}X^\mathsf{T} Y\Big)\right) \nonumber\\
		&\qquad \qquad\qquad\qquad \qquad \cdot  \mathbbm{1}\left\{ \norm{X}^2 > w_{N_1^\sigma},\, \norm{Y}^2 > w_{N_1^\sigma}, \, \sign{\rho}X^\mathsf{T}Y<v_{N_1^\sigma}\right\} \bigg] \nonumber\\
		& \triangleq  (1-\rho^2)^{-dN_1^\sigma}\, \expe_{X,Y\, \iid\, \mathcal{N}(0,\id_{dN_1^\sigma})}	\left[h(X,Y)\mathbbm{1}_{\mathcal{E}_{N_1^\sigma}} \right] \label{eq: converse-truncated-second-product2}
	\end{flalign}

	Finally, by plugging  the results~\eqref{eq: converse-truncated-second-product1} and~\eqref{eq: converse-truncated-second-product2} into Equation~\eqref{eq: converse-truncated-second-likelihood}, the second-order moment of the truncated likelihood ratio is bounded as
	\begin{flalign*}
		\expe_0 \,\tilde{L} ^2 \leq & \left(1-\rho^4\right)^{-\frac{dn}{2}} \expe_{\sigma\sim\mathrm{Unif}(S_n)} \left[\left(\frac{1+\rho^2}{1-\rho^2}\right)^{\frac{dN_1^\sigma}{2}} \, \expe_{X,Y\, \iid\, \mathcal{N}(0,\id_{dN_1^\sigma})}\left[	h(X,Y)\mathbbm{1}_{\mathcal{E}_{N_1^\sigma}}\right]\right]\\
		= & \left(1-\rho^4\right)^{-\frac{dn}{2}} \expe_{\sigma\sim\mathrm{Unif}(S_n)} \left[\left(\frac{1+\rho^2}{1-\rho^2}\right)^{\frac{dN_1^\sigma}{2}} \, \expe_{X,Y\, \iid \, \mathcal{N}(0,\id_{dN_1^\sigma})}\left[	h(X,Y)\mathbbm{1}_{\mathcal{E}_{N_1^\sigma}} \cdot \mathbbm{1}\{ N_1^\sigma \leq k^*\} \right]\right]  \nonumber\\
		&+  \left(1-\rho^4\right)^{-\frac{dn}{2}} \expe_{\sigma\sim\mathrm{Unif}(S_n)} \left[\left(\frac{1+\rho^2}{1-\rho^2}\right)^{\frac{dN_1^\sigma}{2}} \, \expe_{X,Y\, \iid \, \mathcal{N}(0,\id_{dN_1^\sigma})}\left[	h(X,Y)\mathbbm{1}_{\mathcal{E}_{N_1^\sigma}} \cdot \mathbbm{1}\{ N_1^\sigma > k^*\}\right] \right] \nonumber \\
		\leq &  \left(1-\rho^4\right)^{-\frac{dn}{2}}  \expe_{\sigma\sim\mathrm{Unif}(S_n)} \left[\left(\frac{1+\rho^2}{1-\rho^2}\right) ^{\frac{dN_1^\sigma}{2}}\, \expe_{X,Y \,\iid\, \mathcal{N}(0,\id_{dN_1^\sigma})}	\Big[h(X,Y) \, \mathbbm{1}\{ N_1^\sigma \leq k^*\}\Big] \right]\tag{\text{I}}\label{eq: converse-truncated-second-I} \\
		&+  \left(1-\rho^4\right)^{-\frac{dn}{2}} \expe_{\sigma\sim\mathrm{Unif}(S_n)} \left[\left(\frac{1+\rho^2}{1-\rho^2}\right)^{\frac{dN_1^\sigma}{2}} \, \expe_{X,Y\, \iid \, \mathcal{N}(0,\id_{dN_1^\sigma})}\left[	h(X,Y)\mathbbm{1}_{\mathcal{E}_{N_1^\sigma}} \cdot \mathbbm{1}\{ N_1^\sigma > k^*\}\right] \right] \ .\tag{\text{II}}\label{eq: converse-truncated-second-II} 
	\end{flalign*}
	
	The cases \eqref{eq: converse-truncated-second-I} and \eqref{eq: converse-truncated-second-II} will be analyzed separately. Therefore, the aim is converted into finding two functions of $k^*$, call them $f_1(k^*)$ and $f_2(k^*)$, such that for $\frac{\rho^2}{1-\rho^2} < f_1(k^*)$, $\eqref{eq: converse-truncated-second-I}= 1+o(1) $ and for $\frac{\rho^2}{1-\rho^2}  < f_2(k^*)$, $\eqref{eq: converse-truncated-second-II} = o(1)$ both hold. Combining them, we find the desired condition for the converse result:
	$$ \text{For }\rho^2 < \min\{f_1(k^*),\, f_2(k^*)\}, \quad \expe_0 \,\tilde{L} ^2 \leq \text{(I)}+\text{(II)} = 1+o(1) . $$
	In order to find the impossibility boundary of  $\rho^2$ as close to the achievability boundary  as possible, we need to choose $k^*\leq n$ as the maximizer of $\min\{f_1(k^*),\, f_2(k^*)\}$. 
	\subsubsection{Bounding \texorpdfstring{$\pmb{{\eqref{eq: converse-truncated-second-I}=1 + o(1)}}$}{(I)=1+o(1)} }
	Since the $\Gamma_\sigma$ event~\eqref{eq: truncation} is empty for $N_1^{\sigma}<k^*$, the indicator event is applied directly. 
	\begin{flalign*}
		\eqref{eq: converse-truncated-second-I} & = \left(1-\rho^4\right)^{-\frac{dn}{2}} \expe_\sigma \left(\frac{1+\rho^2}{1-\rho^2}\right) ^{dN_1^\sigma/2}\, \expe_{X,Y \,\iid\, \mathcal{N}(0,\id_{dN_1^\sigma})}	\Big[h(X,Y) \, \mathbbm{1}\{ N_1^\sigma \leq k^*\}\Big] \\
		& \stackrel{\text{(i)}}{\leq} \left(1-\rho^4\right)^{-\frac{dn}{2}} \left(\frac{1+\rho^2}{1-\rho^2}\right)^{dk^*/2} \, \expe_{X,Y\, \iid\, \mathcal{N}(0,\id_{dN_1^\sigma})}	\Big[h(X,Y)\Big] \\
		& \stackrel{\text{(ii)}}{=}\left(1-\rho^4\right)^{-\frac{dn}{2}} \left(\frac{1+\rho^2}{1-\rho^2}\right)^{dk^*/2}\\
		& \stackrel{\text{(iii)}}{\leq}\exp\left(\frac{dn}{2}\left(\frac{1}{1-\rho^4}-1\right) + \frac{dk^*}{2}\left( \frac{1+\rho^2}{1-\rho^2}-1\right)\right)\\
		& = \exp\left( \frac{dn}{2} \frac{\rho^4}{1-\rho^4} + dk^*\frac{\rho^2}{1-\rho^2}\right)\\
		&  \leq  \exp\left( \frac{dn}{2} \left(\frac{\rho^2}{1-\rho^2}\right)^2+ dk^*\frac{\rho^2}{1-\rho^2}\right) 
	\end{flalign*}
where (i) is by $ \frac{1+\rho^2}{1-\rho^2}\geq 1$ for any $\rho^2\in[0,1]$, (ii) is from Lemma~\ref{lemma: gaussian-mgf-2} with $a=\frac{\rho^2}{1-\rho^2}$ and $b=\frac{2\rho}{1-\rho^2}$, (iii) is by $\ln x \leq x-1$ for $x>0$.
	Therefore, 
	\begin{align}
		\eqref{eq: converse-truncated-second-I}   = 1+o(1) \quad \text{if} \quad  & \frac{dn}{2} \left(\frac{\rho^2}{1-\rho^2}\right)^2 =o(1)  \text{ and }dk^*\frac{\rho^2}{1-\rho^2}=o(1) \nonumber\\
		\iff \,&  
		\frac{\rho^2}{1-\rho^2}= o\left( \min \left\{ \frac{1}{dk^*}, \, \sqrt{\frac{2}{dn}}\right\} \right) \ . \label{eq: converse-truncated-second-f1}
	\end{align}
	
	\subsubsection{Bounding \texorpdfstring{$\pmb{\eqref{eq: converse-truncated-second-II}= o(1)}$}{(II)=o(1)}}
	First, we observe that since $\norm{X}^2 \geq 0$ almost surely  for any random vector and  $\rho\,\sign{\rho} =|\rho|>0$, the inner expectation gets upper bounded by
	\begin{flalign*}
		\expe_{X,Y \,\iid\, \mathcal{N}(0,\id_{dN_1^\sigma})}\left[h(X,Y)\mathbbm{1}_{\mathcal{E}_{N_1^\sigma}}\right] 
		&= \expe_{X,Y} \left[\exp\left(-\frac{\rho^2}{1-\rho^2}\,\Big(\norm{X}^2+ \norm{Y}^2-\frac{2}{\rho}X^\mathsf{T} Y\big)\right) \right.\\
		&\left.\qquad\qquad\cdot\mathbbm{1}\left\{  \norm{X}^2 > w_{N_1^\sigma},\, \norm{Y}^2 > w_{N_1^\sigma},\,\sign{\rho}X^\mathsf{T}Y<v_{N_1^\sigma}\right\}\right]\\
		& \leq \exp\left(-\frac{2\rho^2}{1-\rho^2} w_{N_1^\sigma}\right) \expe_{X,Y} \Bigg[\exp\Big( \frac{2\rho}{1-\rho^2}X^\mathsf{T} Y   \Big)  \mathbbm{1}\left\{\sign{\rho} X^\mathsf{T}Y<v_{N_1^\sigma}\right\} \Bigg]\\
		& \leq \exp\left(-\frac{2\rho^2}{1-\rho^2} w_{N_1^\sigma} + \frac{2|\rho|}{1-\rho^2} v_{N_1^\sigma} \right) \ .
	\end{flalign*}
	Putting this into $\eqref{eq: converse-truncated-second-II}$, we get
	\begin{flalign}
		\eqref{eq: converse-truncated-second-II} & \stackrel{\hphantom{\text{(iii)}}}{=} \left(1-\rho^4\right)^{-\frac{dn}{2}} \expe_{\sigma\sim\mathrm{Unif}(S_n)}\left[\left(\frac{1+\rho^2}{1-\rho^2}\right)^{\frac{dN_1^\sigma}{2}} \, \expe_{X,Y \,\iid \, \mathcal{N}(0,\id_{dN_1^\sigma})}	\left[h(X,Y)\mathbbm{1}_{\mathcal{E}_{N_1^\sigma}} \, \mathbbm{1}\{ N_1^\sigma > k^*\}\right] \right]\nonumber \\
		& \stackrel{\hphantom{\text{(iii)}}}{\leq} \left(1-\rho^4\right)^{-\frac{dn}{2}} \expe_{\sigma\sim\mathrm{Unif}(S_n)} \left[\left(\frac{1+\rho^2}{1-\rho^2}\right)^{\frac{dN_1^\sigma}{2}}\exp\left(-\frac{2\rho^2}{1-\rho^2} \Big(w_{N_1^\sigma} - \frac{v_{N_1^\sigma}}{|\rho|} \Big) \right) \mathbbm{1}\{ N_1^\sigma > k^*\}\right] \nonumber  \\
		& \stackrel{\hphantom{\text{(iii)}}}{=} \expe_{\sigma\sim\mathrm{Unif}(S_n)} \left[\exp \left( \frac{dn}{2}\ln \frac{1}{1-\rho^4} + \frac{dN_1^\sigma}{2}\ln \frac{1+\rho^2}{1-\rho^2}-\frac{2\rho^2}{1-\rho^2} \Big(w_{N_1^\sigma} - \frac{v_{N_1^\sigma}}{|\rho|}  \,\Big) \right)\mathbbm{1}\{ N_1^\sigma > k^*\}\right] \nonumber  \\
		& \stackrel{\text{\hphantom{i}(i)\hphantom{i}}}{\leq} \sum_{k=k^*}^n\exp\left( \frac{dn}{2} \frac{\rho^4}{1-\rho^4} + dk\frac{\rho^2}{1-\rho^2} -\frac{2\rho^2}{1-\rho^2} \Big(w_k - \frac{v_k}{|\rho|} \, \Big) \right) \mathbb{P}\left[ N_1^\sigma = k\right]   \nonumber  \\
		& \stackrel{\,\text{(ii)}\,}{\leq} \sum_{k=k^*}^n\exp\left( -k \left(-\frac{dn}{2k} \frac{\rho^4}{1-\rho^4} - d\frac{\rho^2}{1-\rho^2} +\frac{2\rho^2}{1-\rho^2} \Big(\frac{w_k }{k}- \frac{v_k}{k|\rho|} \, \Big) + \ln \frac{k}{e}\right) \right) \nonumber  \\
		& \stackrel{\text{(iii)}}{\leq} \sum_{k=k^*}^n e^{-k\psi}  \nonumber\\
		&\stackrel{\hphantom{\text{(iv)}}}{\leq} e^{-k^*\psi }\sum_{k=0}^{\infty} e^{-k\psi} \nonumber\\
		&\stackrel{\text{(iv)}}{\leq} \frac{e^{-k^*\psi}}{1-e^{-\psi}}\label{eq: converse-truncated-second-II-eq}
	\end{flalign}
	where  (i) follows from $\ln x \leq x-1 $ for $x>0$, (ii) is by $\mathbb{P}\left[ N_1^\sigma = k\right]\leq\frac{1}{ k!} \leq \left(\frac{e}{k}\right)^k $ from Stirling's approximation, in (iii), $\psi$ is defined as
		\begin{equation}\label{eq: converse-truncated-second-min}
		\psi = \min_{k=k^*,\ldots, n}  -\frac{dn}{2k} \frac{\rho^4}{1-\rho^4} - d\frac{\rho^2}{1-\rho^2} +\frac{2\rho^2}{1-\rho^2} \Big(\frac{w_k }{k}- \frac{v_k}{k|\rho|}  \Big) + \ln \frac{k}{e}\ ,
	\end{equation} and in (iv), $\psi>0$ should hold. 

\paragraph{Choosing the thresholds $\pmb{w_k}$, $\pmb{v_k}$:}
	To make Equation~\eqref{eq: converse-truncated-second-II-eq} as tight as possible, $\psi$  should be as large as possible. Hence, $\frac{w_k }{k}- \frac{v_k}{k|\rho|}$  should be chosen as large as possible for any $k=k^*,\ldots,n$. Recall their definition from Lemma~\ref{lemma: truncated-first}
	$$	\frac{w_k }{k}- \frac{v_k}{k|\rho|}  = \frac{dk-2\sqrt{d} k r_k}{k} - \frac{|\rho| dk +4|\rho| \sqrt{d} k s_k}{k|\rho|} = -2\sqrt{d}(r_k+ 2s_k) 	$$
	and its result $\min\{ \psi_1, \psi_2\}>0$ obligates 
	$$r_k +2s_k > \sqrt{\ln \frac{e n}{k}}\left(1 +2\max\left\{2, \sqrt{\frac{1-\rho^2}{\rho^2}} \right\}  \right)$$
	to hold. Moreover, since $\expe_{0}\, \tilde{L}^2 \leq \expe_{0}\, L^2$, by Lemma~\ref{lemma: converse-unconditional} we need to have $\rho^2 \leq \frac{1}{dn}$ for $\expe_{0}\, L^2\leq1+o(1)$. Therefore, for $dn\geq 5$, 
	$\max\left\{2, \sqrt{\frac{1-\rho^2}{\rho^2}}\right\} =  \sqrt{\frac{1-\rho^2}{\rho^2}}$. 
	
	It is sufficient to fix $\frac{w_k }{k}- \frac{v_k}{k|\rho|}$ as 
	\begin{equation}\label{eq: converse-truncated-wkvk}
		\frac{w_k }{k}- \frac{v_k}{k|\rho|}  \triangleq -5\sqrt{d\ln \frac{e n}{k}\vphantom{\frac{\rho^2}{\rho^2}}}\sqrt{\frac{1-\rho^2}{\rho^2}}  \ .
	\end{equation}
	
\paragraph{Condition $\pmb{\psi>0}$:}
	By placing these chosen thresholds in~\eqref{eq: converse-truncated-wkvk} into the definition of $\psi$ in~\eqref{eq: converse-truncated-second-min}, the condition $\psi>0$ in the result~\eqref{eq: converse-truncated-second-II-eq} is satisfied if 
\begin{align*}
	\psi>0\Longleftarrow\,	&\ln \frac{k}{e}> \frac{dn}{2k} \frac{\rho^4}{1-\rho^4} + d\frac{\rho^2}{1-\rho^2} +  10\sqrt{\frac{\rho^2}{1-\rho^2}}\sqrt{d\ln \frac{e n}{k}\vphantom{\frac{\rho^2}{\rho^2}}} \quad \forall k=k^*,\ldots,n \ . \nonumber
	\intertext{Since LHS is increasing in $k$ and RHS is decreasing in $k$, the last expression is equivalent to}
		\iff &\ln \frac{k^*}{e} > \frac{dn}{\;2k^*} \frac{\rho^4}{1-\rho^4} + d\frac{\rho^2}{1-\rho^2} + 10\sqrt{\frac{\rho^2}{1-\rho^2}}\sqrt{d\ln \frac{e n}{\;k^*}\vphantom{\frac{\rho^2}{\rho^2}}} \ . \nonumber
	\intertext{By defining $\xi \triangleq\sqrt{\frac{\rho^2}{1-\rho^2}}$ for brevity and by using the inequality $\frac{\rho^4}{1-\rho^4} \leq \frac{\rho^2}{1-\rho^2}$, this is implied by}
	\Longleftarrow \;\; &\underbrace{ \frac{1}{d}\ln \frac{k^*}{e}}_{\triangleq \, c}  > \xi^2\underbrace{ \left(\frac{n}{\;2k^*} +1\right)}_{\triangleq\, a} +  2\xi\cdot\underbrace{ 5 \sqrt{\frac{1}{d}\ln \frac{e n}{\;k^*}}}_{\triangleq\, b} 
	\intertext{By completing the square of a second order polynomial with variable $z$, this is equivalent to }
	\iff & \xi < \frac{b}{a}\left(-1 + \sqrt{1+ \frac{c\cdot a}{b^2}}\,\right) \\
	\stackrel{\text{ (i)}}{\Longleftarrow}\;\;&  \xi < \frac{1}{\sqrt{26}+5}\,\sqrt{\frac{c}{a}} \quad \text{ for } \; \frac{c\cdot a}{b^2} \geq \frac{1}{25}\\
	\stackrel{\text{ (ii)}}{\Longleftarrow}\;\;&   \xi^2  \leq \frac{1}{169}\frac{k^*}{dn} \quad \text{ for } \; k^*\geq e\sqrt{n}
\end{align*}
where (i) comes from  Lemma~\ref{lemma: ineq_sqrt}.
For (ii), by putting back the definitions of $a,b$ and $c$, the condition $\frac{c\cdot a}{b^2} \geq \frac{1}{25}$ holds for $k^*\geq e \sqrt{n}$ as 
$$	\frac{c\cdot a}{b^2} =\frac{\frac{1}{d}\ln \frac{k^*}{e}  \left(\frac{n}{\;2k^*} +1\right)}{25 \frac{1}{d}\ln \frac{e n}{\;k^*}} \geq  \frac{\ln \frac{k^*}{e} }{25 \ln \frac{e n}{\;k^*}} \geq \frac{1}{25} $$
where the inequality
$\ln \frac{k^*}{e} / { \ln \frac{e n}{\;k^*}} \geq 1$ for $k^*\geq e \sqrt{n}$ is used. 
Finally, by $\frac{n}{2k^*} + 1\leq \frac{3n}{\;2k^*}$ resulting from $k^*\leq n$, (ii) is  completed. Note that for $e\sqrt{n} \leq k^* \leq n$ to hold,  $n \geq e^2$ is needed.
	
\paragraph{Summing up for $\pmb{\eqref{eq: converse-truncated-second-II}=o(1)}$:} Since the conditions $\psi>0$ and $k^*\to \infty$ are needed for  \eqref{eq: converse-truncated-second-II-eq}$\to 0$, by putting back $\xi^2 = \frac{\rho^2}{1-\rho^2}$, $	\eqref{eq: converse-truncated-second-II}=o(1)$ holds if
	\begin{equation}\label{eq: converse-truncated-second-f2}
	\frac{\rho^2}{1-\rho^2} =
 o\left( \frac{\,k^*}{dn} \right) \quad  \text{for } \; k^* \geq e\sqrt{n}\ .
\end{equation}
	
\subsubsection{Choosing \texorpdfstring{$\pmb{k^*}$}{k*}}
	We first observe that for $k^*\geq e\sqrt{n}$, we have
	$$\frac{1}{\sqrt{d}}\min \left\{ \frac{1}{\sqrt{d}k^*}, \, \frac{\sqrt{2}}{\sqrt{n}}\right\} = \frac{1}{\;dk^*} \ .$$
	
	We now discover that $\displaystyle f_1(k^*)=\frac{1}{\;dk^*}$ in~\eqref{eq: converse-truncated-second-f1} is a decreasing function of $k^*$ and $\displaystyle f_2(k^*)=\frac{1}{169}\frac{k^*}{dn}$  in~\eqref{eq: converse-truncated-second-f2} is an increasing 
	function of $k^*$. Due to this trade-off between the functions, $\displaystyle \argmax_{1\leq k^* \leq n}\, \min\{f_1(k^*),\, f_2(k^*)\} =\{k^* : \, f_1(k^*)-f_2(k^*)=0\}$. 
		Hence, the tightest converse regime that we can find is
	\begin{align*}
		\frac{\rho^2}{1-\rho^2} &= \max_{k^*} \,\min \big\{ o\big(f_1(k^*)\big),o\big(f_2(k^*)\big)\big\}\\
		& =  o\left( \,\max_{e\sqrt{n}\leq k^*} \min \left\{\frac{1}{\; d k^*} ,\; \frac{1}{169}\frac{\,k^*}{dn}\right\}  \right)\\
		& = o\left(\,\max_{e\sqrt{n}\leq k^* \leq n}\, \frac{1}{\;dk^*} \min \left\{ 1,\; \frac{1}{169}\frac{(k^*)^2}{n}\right\} \right)\\
		& =  o\left( \,\frac{1}{13}\frac{1}{d\sqrt{n}}\, \right) 
	\end{align*}
where in the last step, $k^*=13\sqrt{n}>e\sqrt{n}$ is chosen.
Therefore,
$$ \text{ if } \; \rho^2=o\left( \frac{1}{d\sqrt{n}} \right) ,  \quad\text{then } \; \expe_0 \,\tilde{L} ^2\leq\eqref{eq: converse-truncated-second-I} + \eqref{eq: converse-truncated-second-II}=1+o(1) $$
holds and we get the final result.
\end{proof}

%% file: app-Ineqs.tex
\section{Technical Inequalities}
\begin{lemma}\label{lemma: minimize-g}
	Let $d\geq 0$, $a \in \RR$ and define $\gamma = \left(\frac{2a}{d}\right)^2$. Then, 
	$$\min_{|x|\leq 1}\,  ax + \frac{d}{2}\ln \frac{1}{1-x^2}= \frac{d}{2}\left(\ln \frac{\sqrt{\gamma +1}+1}{2} +1 - \sqrt{\gamma +1}\, \right) $$
	with the minimizer at $x = \frac{d}{2a} - \sign{a}\sqrt{1+\left(\frac{d}{2a}\right)^2}$.
\end{lemma}
\begin{proof}
The proof follows by simply analyzing the first and second derivatives of the given function. Define $g(x)\triangleq \frac{d}{2}\ln \frac{1}{1-x^2} +ax$ over $|x| \leq 1$. Then, we have $g'(x) = d \frac{x}{1-x^2}+a$ and $g''(x) = d\frac{1+x^2}{(1-x^2)^2}$. Clearly, for any $x$, $g''(x)\geq 0 $ if $d\geq 0$ meaning that the function is convex. By finding $x^*=\frac{d}{2a}-\sign{a}\sqrt{1+\big(\frac{d}{2a}\big)^2} \in [-1,1]$ such that $g'(x^*)=0$ holds and plugging it in $g(x^*)$, we get  the desired result.
\end{proof}

\begin{lemma}\label{lemma: ach-f-lemma}
	For $x\in[0,1]$,
	$$x^2(\sqrt{1+x}-1)\leq \sqrt{1+x^3}-1\leq x\,(\sqrt{1+x}-1) \ .$$
\end{lemma}
\begin{proof}
	We begin with the inequality on RHS. Instead of proving it directly, we add a middle step as 
	$$0 \leq \sqrt{1+x^3}-1\stackrel{\text{(i)}}{\leq} x^2(\sqrt{2}-1) \stackrel{\text{(ii)}}{\leq} x\,(\sqrt{1+x}-1) \ .$$
	By taking the square of both sides, (i) becomes 
	\begin{align*}
		\text{(i)}&\iff \sqrt{1+x^3}\leq x^2(\sqrt{2}-1)+1   \iff  x\leq (\sqrt{2}-1)^2x^2 +2(\sqrt{2}-1) \\
		&\iff 0 \leq ((\sqrt{2}-1)x-2)((\sqrt{2}-1)x - (\sqrt{2}-1))
	\end{align*}
	and upon examining the roots, one can easily see that this holds true for $x \in [0,1]$.  On the other hand, after a small manipulation, (ii) is equivalent to $x(\sqrt{2}-1) +1 \leq \sqrt{1+x}$ and this is true for the given range of $x$ due to concavity of square root function. Hence, RHS of the statement follows by combining (i) and (ii).
	
	To show LHS, we take a different approach. First, we observe that the both functions $x^4(\sqrt{1+x}-1)$ and $x^2(\sqrt{1+x^3}-1)$ carry the same value at the boundary points and they are both increasing. Moreover, their second derivatives are nonnegative within the given interval. Hence, their concavity does not change which implies that those two functions do not cross each other. Therefore, we can argue the following:
	\begin{align*}
		x^2(\sqrt{1+x}-1)\leq \sqrt{1+x^3}-1 & \iff  x^4(\sqrt{1+x}-1)\leq x^2(\sqrt{1+x^3}-1) \\
		& \iff \int_0^1 x^4(\sqrt{1+x}-1)\dif x \leq \int_0^1 x^2(\sqrt{1+x^3}-1)\dif x
	\end{align*}
	where by the change of variables, the first integral becomes
	$$\int_0^1 x^4(\sqrt{1+x}-1)\dif x = \int_1^2 (u-1)^4 \sqrt{u}\dif u -\int_0^1x^4 \dif x   = \int_1^2 (u^{9/2} - 4u^{7/2}+6u^{5/2}-4^{3/2}+\sqrt{u})\dif u - \frac{x^5}{5} \Big|_{0}^1  \leq 0.0706$$
	and the second one is calculated easily
	$$\int_0^1 x^2(\sqrt{1+x^3}-1)\dif x = \frac{2}{9}(1+x^3)^{3/2}\Big|_0^1-\frac{x^3}{3} \Big|_{0}^1 \geq 0.0723 \ .$$
	By comparing these two values, the inequality on LHS follows.
\end{proof}

\begin{lemma}\label{lemma: ach-f-func}
	For $x \in (0,1]$,
	$$\ln \, \frac{\sqrt{1-x+x^2} +1-x}{\sqrt{1+x} +1} +\frac{1}{1-x}\left(x - \sqrt{1-x+x^2}\right) +\sqrt{1+x}\leq (\sqrt{2}-1)^2x \ .$$
\end{lemma}
\begin{proof}
	Let 
	$$h(x) \coloneqq \ln \, \frac{\sqrt{1-x+x^2} +1-x}{\sqrt{1+x} +1}, \quad g(x) \coloneqq \frac{1}{1-x}\left(x - \sqrt{1-x+x^2}\right)+\sqrt{1+x} \ .$$ We will bound the functions separately by showing that $h(x)\leq -2(\sqrt{2}-1)x$ and $g(x)\leq x$ hold or $x\in (0,1]$. 
	
	Regarding $h(x)$, by inequality $\sqrt{1+x^3}-1\leq x\,(\sqrt{1+x}-1)$ in Lemma~\ref{lemma: ach-f-lemma} we observe that
	\begin{align*}
		\frac{\sqrt{1-x+x^2} +1-x}{\sqrt{1+x} +1} = \frac{\sqrt{1+x^3} +(1-x)\sqrt{1+x}}{\sqrt{1+x}\, \left(\sqrt{1+x} +1\right)} \leq \frac{1-x+\sqrt{1+x}}{\sqrt{1+x}\, \left(\sqrt{1+x} +1\right)} \ .
	\end{align*}
	By using $\ln y \leq 2 - \frac{4}{y+1}$ for $y \in (0,1]$ (See \cite{log-pade} for more details.) and the previous result, we get
	\begin{align*}
		\frac{h(x)}{2} & \leq \half \ln \frac{1-x+\sqrt{1+x}}{\sqrt{1+x}\, \left(\sqrt{1+x} +1\right)} \leq 1 -\frac{2\sqrt{1+x}\, \left(\sqrt{1+x} +1\right)}{1-x+\sqrt{1+x} + \sqrt{1+x}\, \left(\sqrt{1+x} +1\right)}  \\
		& = 1-\sqrt{1+x} \leq -(\sqrt{2}-1)x
	\end{align*}
	where the last step is due to concavity of square root.
	
	For $g(x)$, we apply the  inequality $x^2(\sqrt{1+x}-1)\leq \sqrt{1+x^3}-1$ in Lemma~\ref{lemma: ach-f-lemma} as
	\begin{align*}
		g(x) = \frac{x\sqrt{1+x}-\sqrt{1+x^3} + 1-x^2}{(1-x)\sqrt{1+x}} \leq  \frac{x\sqrt{1+x}-x^2(\sqrt{1+x}-1)  -x^2}{(1-x)\sqrt{1+x}} = x \ . 
	\end{align*}
	Hence, by summing $h(x)+g(x)\leq -2(\sqrt{2}-1)x+x = (\sqrt{2}-1)^2x$, the result follows.
\end{proof}

\begin{lemma}\label{lemma: ineq_sqrt}
	Let $c\in [0,1]$ and $x_0\geq1$. For $x\geq x_0^2-1$, the inequality 
	$-1 + \sqrt{1+x}\geq c\sqrt{x}$
	holds, if $c \leq \sqrt{\frac{x_0-1}{x_0+1}}$. 
\end{lemma}
\begin{proof}
By completing the square of a second order polynomial with variable $\sqrt{x+1}$ and analyzing its roots, we get the following equivalence relation for $\sqrt{x+1}\geq x_0$:
\begin{align*}
	-1+\sqrt{1+x}\geq c\sqrt{x} \stackrel{\mathrm{(i)}}{\iff}& (1-c^2)(x+1)-2\sqrt{x+1}+1+c^2 \geq 0\\
	\stackrel{\mathrm{(ii)}}{\iff} & \left(\sqrt{x+1}\right)^2 -\frac{2}{1-c^2}\sqrt{x+1}+\frac{1+c^2}{1-c^2} \geq 0 \\
	\stackrel{\mathrm{(iii)}}{\Longleftarrow}\;\;&\left(\sqrt{x+1}-x_0\right)\left(\sqrt{x+1}-\frac{1}{x_0}\frac{1+c^2}{1-c^2}\right) \geq 0
\end{align*}
where in (i)  $c\geq 0$, in (ii) $c\leq 1$ are used and (iii) holds for $x_0 \geq \frac{1}{x_0}\frac{1+c^2}{1-c^2}$ and $x_0 + \frac{1}{x_0}\frac{1+c^2}{1-c^2} \geq \frac{2}{1-c^2}$ which implies $c \leq \sqrt{\frac{x_0-1}{x_0+1}}$ for $x_0\geq 1$.
\end{proof}